 \documentclass[final,1p,times]{elsarticle}

\usepackage{graphicx}
\usepackage{algorithm}
\usepackage{algorithmic}
\usepackage{dsfont}
\usepackage{amssymb,amsmath}
\usepackage[usenames,dvipsnames]{color}
\usepackage{enumerate}
\usepackage{ifthen}
\usepackage{subfigure}
\graphicspath{{figures/}}

\def\conv{*}
\def\P{\mathbb{P}}

\def\eps{\varepsilon}

\begin{document}

\begin{frontmatter}

%% Title, authors and addresses

%% use the tnoteref command within \title for footnotes;
%% use the tnotetext command for the associated footnote;
%% use the fnref command within \author or \address for footnotes;
%% use the fntext command for the associated footnote;
%% use the corref command within \author for corresponding author footnotes;
%% use the cortext command for the associated footnote;
%% use the ead command for the email address,
%% and the form \ead[url] for the home page:
%%
%% \title{Title\tnoteref{label1}}
%% \tnotetext[label1]{}
%% \author{Name\corref{cor1}\fnref{label2}}
%% \ead{email address}
%% \ead[url]{home page}
%% \fntext[label2]{}
%% \cortext[cor1]{}
%% \address{Address\fnref{label3}}
%% \fntext[label3]{}

\newtheorem{theorem}{Theorem}
\newtheorem{proposition}{Proposition}
\newtheorem{lemma}{Lemma}
\newtheorem{corollary}{Corollary}
\newtheorem{proof}{Proof}

\title{Characterizing the Impact of the Workload on \\the Value of Dynamic Resizing in Data Centers\tnoteref{t1}}
\tnotetext[t1]{This technical report is the full version of an accepted poster at ACM Sigmetrics/Performance 2012~\cite{Wang2012} and an accepted paper at the IEEE Infocom 2013 Mini Conference~\cite{Wang2013}}.

%% use optional labels to link authors explicitly to addresses:
%% \author[label1,label2]{<author name>}
%% \address[label1]{<address>}
%% \address[label2]{<address>}

\author[1]{Kai Wang\corref{cor1}}
\ead{wangkai@iscas.ac.cn}
\author[2]{Minghong~Lin}
\ead{mhlin@caltech.edu}
\author[3]{Florin~Ciucu}
\ead{florin@dcs.warwick.ac.uk}
\author[4]{Adam~Wierman}
\ead{adamw@caltech.edu}
\author[5]{Chuang~Lin}
\ead{chlin@tsinghua.edu.cn}

%\cortext[cor1]{Corresponding author. Tel.: +86 138 1049 1901.}

\address[1]{Institute of Software, Chinese Academy of Sciences, China}% S. 4th Street, Zhongguancun, Beijing, 100190 China}
\address[2]{Facebook, USA}
\address[3]{Computer Science Department, University of Warwick, UK} %Coventry, CV4 7AL, UK}
\address[4]{Dept. of Computing and Mathematical Sciences, California Institute of Technology, USA}% Pasadena, CA, 91125 USA}
\address[5]{Dept. of Computer Science and Technology, Tsinghua University, China}%Beijing, 100084 China}

\begin{abstract}
Energy consumption imposes a significant cost for data centers; yet much of that energy is used to maintain excess service capacity during periods of predictably low load. Resultantly, there has recently been interest in developing designs that allow the service capacity to be dynamically resized to match the current workload.  However, there is still much debate about the value of such approaches in real settings.  In this paper, we show that the value of dynamic resizing is highly dependent on statistics of the workload process.  In particular, both slow time-scale non-stationarities of the workload (e.g., the peak-to-mean ratio) and the fast time-scale stochasticity (e.g., the burstiness of arrivals) play key roles.  To illustrate the impact of these factors, we combine optimization-based modeling of the slow time-scale with stochastic modeling of the fast time scale.  Within this framework, we provide both analytic and numerical results characterizing when dynamic resizing does (and does not) provide benefits.
\end{abstract}

\begin{keyword}
%% keywords here, in the form: keyword \sep keyword
Data Centers \sep Dynamic Resizing \sep Energy Efficient IT \sep Stochastic Network Calculus
\end{keyword}

\end{frontmatter}

%%
%% Start line numbering here if you want
%%
% \linenumbers

%% main text
\section{Introduction}

Energy costs represent a significant, and growing, fraction of a data center's budget. Hence there is a push to improve the energy efficiency of data centers, both in terms of the components (servers, disks, network, power infrastructure \cite{MGW09,BH07,Carrera03,Diniz07,Govindan2009}) and the algorithms \cite{A10,Chen08,Chase01,LinWAT11,Lu2013}.  One specific aspect of data center design that is the focus of this paper is dynamically resizing the service capacity of the data center so that during periods of low load some servers are allowed to enter a power-saving mode (e.g., go to sleep or shut down).

The potential benefits of dynamic resizing have been a point of debate in the community \cite{Greenberg08,Chen08,Meisner11,Jin2011,Zhang2012}. On one hand, it is clear that, because data centers are far from perfectly energy proportional, significant energy is used to maintain excess capacity during periods of predictably low load when there is a diurnal workload with a high peak-to-mean ratio.  On the other hand, there are also significant costs to dynamically adjusting the number of active servers.  These costs come in terms of the engineering challenges in making this possible \cite{1251223,Thereska09,ACGGKS10}, as well as the latency, energy, and wear-and-tear costs of the actual ``switching'' operations involved \cite{BACFJP08,Chen08,Gandhi11}.

The challenges for dynamic resizing highlighted above have been the subject of significant research.  At this point, many of the engineering challenges associated with facilitating dynamic resizing have been resolved, e.g., \cite{1251223,Thereska09,ACGGKS10}.  Additionally, the algorithmic challenge of deciding, without knowledge of the future workload, whether to incur the significant ``switching costs'' associated with changing the available service capacity has been studied in depth and a number of promising algorithms have emerged \cite{LinWAT11,Ahmad10,Chen08,Gandhi10,Lu2013}.

However, despite this body of work, the question of characterizing the potential benefits of dynamic resizing has still not been properly addressed.  Providing new insight into this topic is the goal of the current paper.

The perspective of this paper is that, apart from engineering challenges, the key determinant of whether dynamic resizing is valuable is the workload, and that proponents on different sides tend to have different assumptions in this regard.  In particular, a key observation, which is the starting point for our work, is that there are two factors of the workload which provide dynamic resizing potential savings:
\begin{enumerate}[(i)]
\item Non-stationarities at a slow time-scale, e.g., diurnal workload variations.
%\vspace{-.1in}
\item Stochastic variability at a fast time-scale, e.g., the burstiness of request arrivals.
\end{enumerate}
The goal of this work is to investigate the impact of and interaction between these two features with respect to dynamic resizing.

To this point, \emph{we are not aware of any work characterizing the benefits of dynamic resizing that captures both of these features}.  There is one body of literature which provides algorithms that take advantage of (i), e.g., \cite{Chen08,Chase01,LinWAT11,Ahmad10,Chen2010,Verma2009,Zhang2012}. This work tends to use an optimization-based approach to develop dynamic resizing algorithms. There is another body of literature which provides algorithms that take advantage of (ii), e.g., \cite{Gandhi10,Gandhi11}.  This work tends to assume a stationary queueing model with Poisson arrivals to develop dynamic resizing algorithms.

The first contribution of the current paper is to provide an analytic framework that captures both effects (i) and (ii).  We accomplish this by using an optimization framework at the slow time-scale (see Section \ref{s.slow}), which is similar to that of \cite{LinWAT11}, and combining this with stochastic network calculus and large deviations modeling for the fast time-scale (see Section \ref{s.fast}), which allows us to study a wide variety of underlying arrival processes.  We consider
both light-tailed models with various degrees of burstiness and heavy-tailed models that exhibit self-similarity.  The interface between the fast and slow time-scale models happens through a constraint in the optimization problem that captures the Service Level Agreement (SLA) for the data center, which is used by the slow time-scale model but calculated using the fast time-scale model (see Section \ref{s.fast}).

Using this modeling framework, we are able to provide both analytic and numerical results that yield new insight into the potential benefits of dynamic resizing (see Section \ref{s.results}).  Specifically, we use trace-driven numerical simulations to study (i) the role of burstiness for dynamic resizing, (ii) the role of the peak-to-mean ratio for dynamic resizing, (iii) the role of the SLA for dynamic resizing, and (iv) the interaction between (i), (ii), and (iii).  The key realization is that each of these parameters are extremely important for determining the value of dynamic resizing.  In particular, for any fixed choices of two of these parameters, the third can be chosen so that dynamic resizing does or does not provide significant cost savings for the data center. Thus, performing a detailed study of the interaction of these factors is important.  To that end, Figures \ref{fig:whenvaluablebysavings}-\ref{fig:whenvaluablebyepsilon} provide concrete illustrations of which settings of peak-to-mean ratio, burstiness, and SLAs dynamic resizing is and is not valuable.  Hence, debate about the \emph{potential} value of dynamic resizing can be transformed into debate about characteristics of the workload and the SLA.

There are some interesting facts about these parameters individually that our case studies uncover.  Two important examples are the following.  First, while one might expect that increased burstiness provides increased opportunities for dynamic resizing, it turns out the burstiness at the fast time-scale actually reduces the potential cost savings achievable via dynamic resizing.  The reason is that dynamic resizing necessarily happens at the slow time-scale, and so the increased burstiness at the fast time-scale actually results in the SLA constraint requiring \emph{more} servers be used at the slow time-scale due to the possibility of a large burst occurring.  Second, it turns out the impact of the SLA can be quite different depending on whether the arrival process is heavy- or light-tailed.  In particular, as the SLA becomes more strict, the cost savings possible via dynamic resizing under heavy-tailed arrivals decreases quickly; however, the cost savings possible via dynamic resizing under light-tailed workloads is unchanged.

In addition to detailed case studies, we provide analytic results that support many of the insights provided by the numerics.  In particular, Theorems \ref{t.ppscaling} and \ref{t.htscaling} provide monotonicity and scaling results for dynamic resizing in the case of Poisson arrivals and heavy-tailed, self-similar arrivals.

The remainder of the paper is organized as follows. The model is introduced in Sections \ref{s.slow} and \ref{s.fast}, where Section \ref{s.slow} introduces the optimization model of the slow time-scale and Section \ref{s.fast} introduces the model of the fast time-scale and analyzes the impact different arrival models have on the SLA constraint of the dynamic resizing algorithm used in the slow time-scale. Then, Section \ref{s.results} provides case studies and analytic results characterizing the impact of the workload on the benefits of dynamic resizing. The related proofs are presented in Section \ref{s.proofs}. 
Finally, Section \ref{s.conclusion} provides concluding remarks.

\section{Slow Time-scale Model}
\label{s.slow}

In this section and the one that follows, we introduce our model.
We start with the ``slow time-scale model''. This model is meant
to capture what is happening at the time-scale of the data center
control decisions, i.e., at the time-scale which the data center
is willing to adjust its service capacity.  For many reasons, this
is a much slower time-scale than the time-scale at which requests
arrive to the data center.  We provide a model for this ``fast
time-scale" in the next section.

The slow time-scale model parallels closely the model studied
in~\cite{LinWAT11}. The only significant change is to add a
constraint capturing the SLA to the cost optimization solved by
the data center. This is a key change, which allows an interface
to the fast time-scale model.

\subsection{The Workload}

At this time-scale, our goal is to provide a model which can
capture the impact of diurnal non-stationarities in the workload.
To this end, we consider a discrete-time model such that there is
a time interval of interest which is evenly divided into ``frames"
$k \in \left\{1, ..., K \right\}$. In practice, the length of a
frame could be on the order of $5$-$10$ minutes, whereas the time
interval of interest could be as long as a month/year. The mean
request arrival rate to the data center in frame $k$ is denoted by
$\lambda_{k}$, and non-stationarities are captured by allowing
different rates during different frames. Although we could allow
$\lambda_{k}$ to have a vector value to represent more than one
type of workload as long as the resulting cost function is convex
in our model, we assume $\lambda_{k}$ to have a scalar value in
this paper to simplify the presentation.  Because the request
inter-arrival times are much shorter than the frame length,
typically in the order of $1$-$10$ seconds, capacity provisioning
can be based on the average arrival rate during a frame.

\subsection{The Data Center Cost Model}

The model for data center costs focuses on the server costs of the
data center, as minimizing server energy consumption also reduces
cooling and power distribution costs. We model the cost of a
server by the operating costs incurred by an active server, as
well as the switching cost incurred to toggle a server into and
out of a power-saving model (e.g., off/on or sleeping/waking).
Both components can be assumed to include energy cost, delay cost,
and wear-and-tear cost. The model framework we adopt is fairly standard and has been used in a number of previous papers, e.g., see~\cite{LinWAT11, Lin12} for a further discussion of the work and~\cite{Liu2012} for a discussion of how the model relates to implementation challenges.

%See~\cite{LinWAT11} and~\cite{Lin12} for further discussion of the model.

Note that this model ignores many issues surrounding reliability
and availability, which are key components of data center service
level agreements (SLAs). In practice, a solution that toggles
servers must still maintain the reliability and availability
guarantees; however this is beyond the scope of the current paper.
See \cite{Thereska09} for a discussion.

\subsubsection*{The Operating Cost}

The operating costs are modeled by a convex function
$f(\lambda_{i,k})$, which is the same for all the servers, where
$\lambda_{i,k}$ denotes the average arrival rate to server $i$
during frame $k$. The convexity assumption is quite general and
captures many common server models. One example, which we consider
in our numeric examples later, is to say that the operating costs
are simply equal to the energy cost of the server, i.e., the
energy cost of an active server handling arrival rate
$\lambda_{i,k}$. This cost is often modeled using an affine
function as follows
\begin{equation}
f(\lambda_{i,k})=e_{0}+e_{1} \lambda_{i,k}~,
\label{e.operatingcostmodel}
\end{equation}
where $e_{0}$ and $e_{1}$ are constants \cite{spec,BH07,Kusic08}.  Note that
when servers use dynamic speed scaling, if the energy cost is
modeled as polynomial in the chosen speed, the cost $f(\cdot)$
remains convex.  In practice, we expect that $f(\cdot)$ will be
empirically measured by observing the system over time.

\subsubsection*{The Switching Cost}

The switching cost, denoted by $\beta$, models the cost of
toggling a server back-and-forth between active and power-saving
models. The switching cost includes the costs of the energy used
toggling a server, the delay in migrating connections/data when
toggling a server, and the increased wear-and-tear on the servers
toggling.

\subsection{The Data Center Optimization}

Given the cost model above, the data center has two control
decisions at each time: determining $n_{k}$, the number of active
servers in every time frame, and assigning arriving jobs to servers,
i.e., determining $\lambda_{i,k}$ such that $\Sigma_{i=1}^{n_{k}}
\lambda_{i,k}=\lambda_{k}$. All servers are assumed to be
homogeneous with constant rate capacity $\mu>0$. Modeling heterogeneous
servers is possible but the online problem will become more complicated \cite{Lin12} ,
which is out of the scope of this paper.

%Note that the assumptions on homogeneous workload and the servers are made for simplicitity, which are typical to analytical work in this area. They should not have a significant impact on the qualitative conclusions of the work.

The goal of the data center is to determine $n_{k}$ and
$\lambda_{i,k}$ to minimize the cost incurred during $[0,K]$,
which is modeled as follows:

\begin{equation}\label{opt}
\min  \sum\limits_{k=1}^{K} \sum\limits_{i=1}^{n_{k}}
f(\lambda_{i,k})+\beta \sum\limits_{k=1}^{K} (n_{k}-n_{k-1})^{+}
\end{equation}
\begin{equation}\label{s.t.}
s.t. \left\{\
  \begin{aligned}%[\relax][c]{l's}
  0 \leq \lambda_{i,k} \leq \lambda_{k}\\
 \Sigma_{i=1}^{n_{k}} \lambda_{i,k}=\lambda_{k}\\
  \mathbb{P} (D_k>\bar{D}) \leq \bar{\varepsilon}~,
  \end{aligned}
\right.
\end{equation}
where the final constraint is introduced to capture the SLA of the
data center. We use $D_k$ to represent the steady-state delay
during frame $k$, and $(\bar{D},\bar{\epsilon})$ to represent an
SLA of the form ``the probability of a delay larger than $\bar{D}$
must be bounded by probability $\bar{\varepsilon}$".

This model generalizes the data center optimization problem
from~\cite{LinWAT11} by accounting for the additional SLA
constraint. The specific values in this constraint are determined
by the stochastic variability at the fast time-scale.
In particular, we derive (for a variety of workload models) a
sufficient constraint $n_k\geq\frac{C_k(\bar{D},\bar{\eps})}{\mu}$
such that
\begin{equation}
n_k\geq\frac{C_k(\bar{D},\bar{\eps})}{\mu} \Longrightarrow
\mathbb{P} (D_k>\bar{D}) \leq
\bar{\varepsilon}~.\label{eq:suffCond}
\end{equation}
Here, $\mu$ is the constant rate capacity of each server and
$C_k(\bar{D},\bar{\eps})$ is to be determined for each considered
arrival model. One should interpret $C_k(\bar{D},\bar{\eps})$ as
the overall effective capacity/bandwidth needed in the data center
such that the SLA delay constraint is satisfied within frame $k$.

Note that the new constraint is only sufficient for the original
SLA constraint. The reason is that $C_k(\bar{D},\bar{\eps})$ will
be computed, in the next section, from upper bounds on the
distribution of the transient delay within a frame. %This means that the new optimization problem can only yield a conservative or over-estimated solution, relative to the optimization problem from Eq.~(\ref{opts}).

With the new constraint, however, the optimization problem in
(\ref{opt})-(\ref{s.t.}) can be considerably simplified. Indeed,
note that $n_{k}$ is fixed during each time frame $k$ and the
remaining optimization for $\lambda_{i,k}$ is convex. Thus, we can
simplify the form of the optimization problem by using the fact
that the optimal dispatching strategy $\lambda_{i,k}^{*}$ is load
balancing, i.e., $\lambda_{1,k}^{*}= \lambda_{2,k}^{*}=\ldots=
\lambda_{k}/n_{k}$. This decouples dispatching $\lambda_{i,k}^{*}$
from capacity planning $n_{k}$, and so
Eqs.~(\ref{opt})-(\ref{s.t.}) become:
\begin{eqnarray}
&&\ \ \ \underline{\text{Data Center Optimization Problem}}\nonumber\\
&&\min  \sum\limits_{k=1}^{K} n_{k} f(\lambda_{k}/n_{k})+\beta \sum\limits_{k=1}^{K} (n_{k}-n_{k-1})^{+} \label{opts} \\
&&\ \ \ \text{s.t. } n_k\geq\frac{C_k(\bar{D},\bar{\eps})}{\mu}~.
\nonumber
\end{eqnarray}

Note that (\ref{opts}) is a convex optimization, since $n_{k}
f(\lambda_{k}/n_{k})$ is the perspective function of the convex
function $f(\cdot)$.

As we have already pointed out, the key difference between the
optimization above, and that of~\cite{LinWAT11}, is the SLA
constraint. However, this constraint plays a key role in the
current paper. It is this constraint that provides a bridge
between the slow time-scale and fast time-scale models.
Specifically, the fast time-scale model uses large deviations and
stochastic network calculus techniques to calculate $C_k(\bar{D},\bar{\eps})$.

\subsection{Algorithms for Dynamic Resizing}

Though the Data Center Optimization Problem described above is
convex, in practice it must be solved \emph{online}, i.e., without
knowledge of the future workload.  Thus, in determining $n_k$, the
algorithm may not have access to the future arrival rates
$\lambda_l$ for $l>k$. This fact makes developing algorithms for
dynamic resizing challenging.  However, progress has been made
recently~\cite{LinWAT11,LWAT11}.

Deriving algorithms for this problem is not the goal of the
current paper. Thus, we make use of a recent algorithm called Lazy
Capacity Provisioning (LCP) \cite{LinWAT11}.  We choose LCP
because of the strong analytic performance guarantees it provides
-- LCP provides cost within a factor of 3 of optimal for any (even
adversarial) workload process.

LCP works as follows. Let $(n_{k, 1}^L,\ldots,n_{k,k}^L)$ be the
solution vector to the following optimization problem
\begin{eqnarray*}
&&\min  \sum\limits_{l=1}^{k} n_{l} f(\lambda_{l}/n_{l})+\beta  \sum\limits_{l=1}^{k} (n_{l}-n_{l-1})^{+}\\
&&\ \ \ \text{s.t. }
n_l\geq\frac{C_l(\bar{D},\bar{\eps})}{\mu}~,~n_0=0~.
\end{eqnarray*}
Similarly, let $(n_{k, 1}^U,\ldots,n_{k,k}^U)$ be the solution
vector to the following optimization problem
\begin{eqnarray*}
&&\min  \sum\limits_{l=1}^{k} n_{l} f(\lambda_{l}/n_{l})+\beta  \sum\limits_{l=1}^{k} (n_{l-1}-n_{l})^{+}\\
&&\ \ \ \text{s.t. }
n_l\geq\frac{C_l(\bar{D},\bar{\eps})}{\mu}~,~n_0=0~.
\end{eqnarray*}
Denote  $(n)_a^b = \max(\min(n,b),a)$ as the projection of $n$
into the closed interval $[a,b]$. Then LCP can be defined using
$n_{k,k}^L$ and $n_{k,k}^U$ as follows.  Informally, LCP stays ``lazily'' between the upper bound $n_{k,k}^U$ and the lower bound $n_{k,k}^L$ in all frames.

\begin{center} \underline{Lazy Capacity Provisioning, LCP} \end{center}
\noindent \emph{Let $n^{LCP}=(n_0^{LCP},\ldots,n_K^{LCP})$ denote the vector of
active servers under LCP.  This vector can be calculated online using the
following forward recurrence relation:}
\begin{eqnarray*}
n_k^{LCP} = \left\{
          \begin{array}{ll}
            0, & \text{$k \leq 0$} \\
            (n_{k-1}^{LCP})^{n^U_{k,k}}_{n^L_{k,k}}, & \text{$1\le k \le K$}~.
          \end{array}
        \right.
\end{eqnarray*}

Note that, in \cite{LinWAT11}, LCP is introduced and analyzed for
the optimization from Eq.~(\ref{opts}) without the SLA constraint.
However, it is easy to see that the algorithm and performance
guarantee extend to our setting. Specifically, the guarantees on
LCP hold in our setting because the SLA constraint can be removed
by defining the operating cost to be $\infty$ instead of $n_{k}
f(\lambda_{k}/n_{k})$ when $n_k<C_k(\bar{D},\bar{\eps})/\mu$.

A last point to highlight about LCP is that, as described, it does not use any predictions about the workload in future frames.  Such information could clearly be beneficial, and can be incorporated into LCP if desired, see \cite{LinWAT11}.

\section{Fast Time-scale Model}
\label{s.fast}

Given the model of the slow time-scale in the previous section, we
now zoom in to give a description for the fast time-scale model.
By ``fast'' time-scale, we mean the time-scale at which requests
arrive, as opposed to the ``slow'' time-scale at which dynamic
resizing decisions are made by the data center. To model the fast
time-scale, we evenly break each frame from the slow time-scale
into ``slots'' $t\in\{1,\ldots,U\}$, such that
$\textrm{frame\_length}=U\cdot\textrm{slot\_length}$.

We consider a variety of models for the workload process at this
fast time-scale, including both light-tailed models with various
degrees of burstiness, as well as heavy-tailed models that exhibit
self-similarity. In all cases, our assumption is that the workload
is stationary over the slots that make up each time frame.

The goal of this section is to derive the value of
$C_k(\bar{D},\bar{\eps})$ in the constraint $n_k\geq
\frac{C_k(\bar{D},\bar{\eps})}{\mu}$ from Eq.~(\ref{eq:suffCond}),
and thus enable an interface between the fast and slow time-scales
by parameterizing the Data Center Optimization Problem from
Eq.~(\ref{opts}) for a broad range of workloads.

Note that throughout this section we suppress frame's subscript
$k$ for $n_k$, $\lambda_k$, $C_k$, and $D_k$, and focus on a
generic frame.

Our approach for deriving the SLA constraint for the Data Center
Optimization Problem will be to first derive an ``aggregation
property" which allows the data center to be modeled as a single
server, and to then derive bounds on the distribution of the
transient delay under a variety of arrival processes.

\subsection{An Aggregation Property}

Note that, if the arrival process were modeled as Poisson and job
sizes were exponential, then an ``aggregation property'' would be
immediate, since the response time distribution only depends on
the load. Hence the SLA could be derived by considering a single
server.  Outside of this simple case, however, we need to derive a
suitable single server approximation.

The aggregation result that we derive and apply is formulated in
the framework of stochastic network calculus~\cite{Book-Chang},
and so we begin by briefly introducing this framework.

Denote the cumulative arrival (workload) process at the data
center's dispatcher by $A(t)$. That is, for each slot
$t=1,\dots,U$, $A(t)$ counts the total number of jobs arrived in
the time interval $[0,t]$. Depending on the total number $n$ of
active servers, the arrival process is dispatched into the
sub-arrival processes $A_i(t)$ with $i=1,\dots,n$ such that
$A(t)=\sum_i A_i(t)$. The cumulative response processes from the
servers are denoted by $R_i(t)$, whereas the total cumulative
response process from the data center is denoted by $R(t)=\sum_i
R_i(t)$. All arrival and response processes are assumed to be
non-negative, non-decreasing, and left-continuous, and satisfy the
initial condition $A(0)=R(0)=0$. For convenience we use the
bivariate extensions $A(s,t):=A(t)-A(s)$ and $R(s,t):=R(t)-R(s)$.

The service provided by a server is modeled in terms of
probabilistic lower bounds using the concept of a stochastic
service process. This is a bivariate random process $S(s,t)$ which
is non-negative, non-decreasing, and left-continuous. Formally, a
server is said to guarantee a (stochastic) service process
$S(s,t)$ if for \textit{any} arrival process $A(t)$ the
corresponding response process $R(t)$ from the server satisfies
for all $t\geq0$
\begin{equation}\label{service}
R(t) \geq A \conv S(t)~,
\end{equation}
where `$\conv$' denotes the min-plus convolution operator, i.e.,
for two (random) processes $A(t)$ and $S(s,t)$,
\begin{equation}
A\conv S (t):=\inf\limits_{0\leq s\leq t}
\left\{A(s)+S(s,t)\right\}~.\label{eq:scdef}
\end{equation}
The inequality in (\ref{service}) is assumed to hold almost
surely. Note that the lower bound set by the service process is
invariant to the arrival processes.

We are now ready to state the aggregation property.  The proof is
deferred to Section~\ref{s.proofs}.

\begin{lemma}\label{aggregate}Consider an arrival process
$A(t)$ which is dispatched to $n$ servers. Each server $i$ is
work-conserving with constant rate capacity $\mu>0$. Arrivals are
dispatched deterministically across the servers such that each
server $i$ receives a fraction $\frac{1}{n}$ of the arrivals.
Then, the system has service process $S(s,t)=n\mu(t-s),$ i.e.,
$R(t)\geq A \conv S(t)$.
\end{lemma}

The significance of the Lemma is that if the SLA is verified for
the virtual server with arrival process $A(t)$ and service process
$S(s,t)$, then the SLA is verified for each of the $n$ servers. We
point out that the lemma is based on the availability of a
dispatching policy with equal weights for homogenous servers.

\subsection{Arrival Processes}

Now that we can reduce the study of the multi-server system to the
study of a single server system using Lemma \ref{aggregate}, we
can move to characterizing the impact of the arrival process on
the SLA constraint in the Data Center Optimization Problem.

In particular, the next step in deriving the SLA constraint $n\geq
\frac{C(\bar{D},\bar{\eps})}{\mu}$ is to derive a bound on the
distribution of the delay at the virtual server with arrival
process $A(t)$ and service process
$S(s,t)=C(\bar{D},\bar{\eps})(t-s)$, i.e.,
\begin{equation}
\P\Big(D(t)>\bar{D}\Big)\leq\eps(\bar{D})~.\label{eq:db0}
\end{equation}
It is important to observe that the violation probability $\eps$
holds for the \textit{transient} \textit{virtual} delay process
$D(t)$, which is defined as $D(t):=\inf\left\{d:A(t-d)\leq
R(t)\right\}$, and which models the delay spent in the system by
the job leaving the system, if any, at time $t$. By this
definition, and using the servers' homogeneity and also the
deterministic splitting of arrivals from Lemma~\ref{aggregate},
the virtual delay for the aggregate virtual server is the same as
the virtual delay for the individual servers. This fact guarantees
that an SLA constraint on the virtual server implicitly holds for
the individual servers as well. Moreover, the violation
probability $\eps$ in Eq.~(\ref{eq:db0}) is derived so that it is
time invariant, which implies that it bounds the distribution of
the stead-state delay $D=\lim_{t\rightarrow\infty}D(t)$ as well.
Therefore, the value of $C(\bar{D},\bar{\eps})$ can be finally
computed by solving the implicit equation
$\eps(\bar{D})=\bar{\eps}$.

In the following, we follow the outline above to compute
$C(\bar{D},\bar{\eps})$ for light- and heavy-tailed arrival
processes.  Figure \ref{fig:trace} depicts examples of the three
types of arrival processes we consider in 1 frame: Poisson,
Markov-Modulated (MM), and heavy-tailed arrivals. In all three
cases, the mean arrival rate is $\lambda=300$. The figure clearly
illustrates the different levels of burstiness of the three
traces.

\begin{figure}[t]
\begin{center}
{\includegraphics[width=0.8\linewidth]{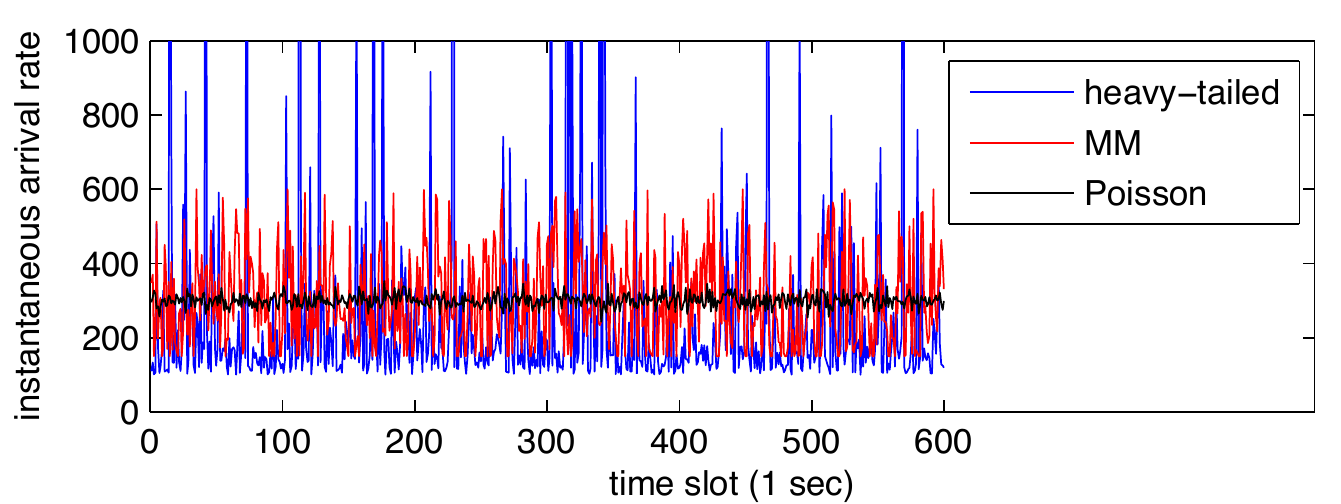}} \caption{Three
synthetically generated traces within 1 frame, with $\lambda=300$
of Poisson, Markov-Modulated (MM) ($T=1$, $\lambda_l=0.5\lambda$
and $\lambda_h=2\lambda$), and heavy-tailed arrivals
($b=\lambda/3$ and $\alpha=1.5$).} \label{fig:trace}
\end{center}
\end{figure}

\subsubsection{Light-tailed Arrivals}
We consider two examples of light-tailed arrival processes:
Poisson and Markov-Modulated (MM) processes.

\subsubsection*{Poisson Arrivals}
We start with the case of Poisson processes, which are
characterized by a low level of burstiness, due to the independent
increments property. The following proposition, providing the tail
of the virtual delay, is a minor variation of a result
from~\cite{Ciucu07}; for the proof see the Appendix.

\begin{proposition}\label{t.Poisson}
Let $A(t)$ be a Poisson process with some rate $\lambda>0$, and
define
\begin{equation}
\theta^*:=\sup\left\{\theta>0:\frac{\lambda}{\theta}\left(e^{\theta}-1\right)\leq
C(\bar{D},\bar{\eps})\right\}~.\label{eq:thetaStar}
\end{equation}
Then a bound on the transient delay process is given for all
$t\geq0$ by
\begin{equation}
\P\left(D(t)>\Bar{D}\right)\leq e^{-\theta^*
C(\Bar{D},\Bar{\varepsilon})\Bar{D}}:=\varepsilon(\bar{D})~.\label{eq:dbMM}
\end{equation}
\end{proposition}

Solving for $C(\bar{D},\bar{\eps})$ by setting the violation
probability $\varepsilon(\bar{D})$ equal to $\bar{\eps}$ yields
the implicit solution
\begin{equation*}
C(\bar{D},\bar{\eps})=-\frac{1}{\theta^*\Bar{D}}\log{\bar{\eps}}~.
\end{equation*}
Further, using the monotonicity of the function
$\frac{\lambda}{\theta}\left(e^{\theta}-1\right)$ in $\theta>0$ we
immediately get the explicit solution
\begin{equation}
C(\bar{D},\bar{\eps})=\frac{K}{\log\left(1+K\right)}\lambda~,\label{eq:CPoisson}
\end{equation}
where $$K=-\frac{\log{\bar{\eps}}}{\lambda\bar{D}}~.$$

%\florin{Kai: note that $\theta^*$ and $C$ depend on each other;
%for fast computation of $C$ you should implement a binary search
%algorithm by dividing the search space in half at each step}
%
%\minghong{Florin \& Kai: From my understanding,
%$\frac{\lambda}{\theta}\left(e^{\theta}-1\right)$ is an increasing
%function of $\theta$, thus
%$\frac{\lambda}{\theta^*}\left(e^{\theta^*}-1\right)=C=-\frac{1}{\theta^*\Bar{D}}\log{\bar{\eps}}$.
%This implies $\theta^*=\log(1-\frac{\log{\bar{\eps}}}{\lambda
%\Bar{D}})$. If we denote $A=-\frac{\log{\bar{\eps}}}{\Bar{D}}$,
%then $C=\frac{A}{\log(1+A/\lambda)}$. Thus $C$ increases as $A$
%increases when \textcolor{Green} {$A/\lambda\ge e-1$} (This is a
%sufficient condition. Does it make sense?). }

\subsubsection*{Markov-Modulated Arrivals}
Consider now the case of Markov-Modulated (MM) processes which,
unlike the Poisson processes, do not necessarily have independent
increments.  The key feature for the purposes of this paper is
that the burstiness of MM processes can be arbitrarily adjusted.

We consider a simple MM processes with two states. Let a discrete
and homogeneous Markov chain $x(s)$ with two states denoted by
`low' and `high', and transition probabilities $p_h$ and $p_l$
between the `low' and `high' states, and vice-versa, respectively.
Assuming that a source produces at some constant rates
$\lambda_l>0$ and $\lambda_h>\lambda_l$ while the chain $x(s)$ is
in the `low' and `high' states, respectively, then the
corresponding MM cumulative arrival process is
\begin{equation}
A(t)=\sum_{s=1}^t\left(
\lambda_lI_{\{x(s)=`low'\}}+\lambda_hI_{\{x(s)=`high'\}}\right)~,\label{eq:mmoop}
\end{equation}
where $I_{\{\cdot\}}$ is the indicator function. The average rate
of $A(t)$ is
$\lambda=\frac{p_l}{p_h+p_l}\lambda_l+\frac{p_h}{p_h+p_l}\lambda_h$.

To adjust the burstiness level of $A(t)$ we introduce the
parameter $T:=\frac{1}{p_h}+\frac{1}{p_l}$, which is the average
time for the Markov chain $x(s)$ to change states twice. We note
that the higher the value of $T$ is, the higher the burstiness
level becomes (the time periods whilst x(s) spends in the `high'
or `low' states get longer and longer).

To compute the delay bound let us construct the matrix
\begin{equation*}
\Psi(\theta)=\left(\begin{array}{cc}(1-p_h)e^{\theta\lambda_l}&p_he^{\theta\lambda_h}\\
p_le^{\theta\lambda_l}&(1-p_l)e^{\theta\lambda_h}\end{array}\right)~,
\end{equation*}
for some $\theta>0$ and consider its spectral radius

\begin{equation}\label{eq:lambdatheta}
\lambda(\theta):=\frac{(1-p_h)e^{\theta
\lambda_l}+(1-p_l)e^{\theta \lambda_h}+\sqrt{\Delta}}{2}~,
\end{equation}
where $\Delta=\left((1-p_h)e^{\theta \lambda_l}-(1-p_l)e^{\theta
\lambda_h}\right)^2+4p_hp_l
e^{\theta\left(\lambda_l+\lambda_h\right)}$. Let also
\begin{equation}\label{eq:ktheta}
K(\theta):=\max\left\{\frac{p_h e^{\theta
\lambda_h}}{\lambda(\theta)-(1-p_h)e^{\theta
\lambda_l}},\frac{\lambda(\theta)-(1-p_h)e^{\theta \lambda_l}}{p_h
e^{\theta \lambda_h}}\right\}~.
\end{equation}
The two terms are the ratios of the elements of the
right-eigenvector of the matrix $\Psi(\theta)$. Also, let
\begin{equation*}
\theta^*:=\sup\left\{\theta>0:\frac{1}{\theta}\log{\lambda(\theta)}\le
C(\bar{D},\bar{\eps})\right\}~.
\end{equation*}

Using these constructions, and also the constant rate service
assumption, a result on the backlog bound from~\cite{Book-Chang},
pp.~340, immediately lends itself to the corresponding result on
the virtual delay:

\begin{proposition}\label{t.mm}
Consider a MM cumulative arrival process as defined
in~(\ref{eq:mmoop}), with the $\lambda(\theta)$ given in
~(\ref{eq:lambdatheta}), and $K(\theta)$ given
in~(\ref{eq:ktheta}), then a bound on the transient delay process
is
\begin{equation*}
\P\Big(D(t)>\Bar{D}\Big)\leq K(\theta^*)e^{-\theta^*
C(\bar{D},\bar{\eps})\Bar{D}}:=\varepsilon(\bar{D})~.
\end{equation*}
\end{proposition}

Setting the violation probability $\varepsilon(\bar{D})$ equal to
$\bar{\eps}$ in Theorem~\ref{t.mm} yields the implicit solution
\begin{equation}
C(\bar{D},\bar{\eps})=-\frac{1}{\theta^*\Bar{D}}\log{\frac{\bar{\eps}}{K(\theta^*)}}~.
\end{equation}

\subsection{Heavy-tailed and Self-similar Arrivals}
We now consider the class of heavy-tailed and self-similar arrival
processes. These processes are fundamentally different from
light-tailed processes in that deviations from the mean increase
in time and decay in probability as a power law, i.e., more slower
than the exponential.

We consider in particular the case of a source generating jobs in
every slot according to i.i.d. Pareto random variables $X_{i}$
with tail distribution for all $x\geq b$:
\begin{equation} \label{e.pareto}
\mathbb{P}\left(X_{i}>x\right)=(x/b)^{-\alpha}~,
\end{equation}
where $1<\alpha<2$. $X$ has finite mean $E[X]=\alpha b/(\alpha-1)$
and infinite variance. For the corresponding bound on the
transient delay we reproduce a result from~\cite{LiBuCi12}.
\begin{proposition}\label{t.htss}
Consider a source generating jobs in every slot according to
i.i.d. Pareto random variables $X_{i}$, with the tail distribution
from~(\ref{e.pareto}). The bound on the transient delay is
\begin{equation} \label{e.HT1}
\P\Big(D(t)>\Bar{D}\Big)\leq
K\left(C(\bar{D},\bar{\eps})\Bar{D}\right)^{1-\alpha}:=\varepsilon(\bar{D})~,
\end{equation}
where
$$K=\inf_{1<\gamma<\frac{C(\bar{D},\bar{\eps})}{\lambda}}\left\{\left(\frac{C(\bar{D},\bar{\eps})}{\gamma}-\lambda\right)^{-1}\frac{\alpha\gamma^{\frac{\alpha-1}{\alpha}}}{(\alpha-1)\log{\gamma}}\right\}.$$
\end{proposition}

Setting the violation probability $\varepsilon(\bar{D})$ equal to
$\bar{\eps}$, we get the implicit solution
\begin{equation}
\inf_{1<\gamma<\frac{C(\bar{D},\bar{\eps})}{\lambda}}\left\{\frac{\gamma}{C(\bar{D},\bar{\eps})^{\alpha-1}\left(C(\bar{D},\bar{\eps})-\gamma\lambda\right)}\frac{\gamma^{\frac{\alpha-1}{\alpha}}}{\log{\gamma^{\frac{\alpha-1}{\alpha}}}}\right\}=\bar{\eps}\bar{D}^{\alpha-1}~.\label{eq:solHT}
\end{equation}

\section{Case Studies}
\label{s.results}

Given the model described in the previous two sections, we are now ready to explore the potential of dynamic resizing in data centers, and how this potential depends on the interaction between non-stationarities at the slow time-scale and burstiness/self-similarity at the fast time-scale.  Our goal in this section is to provide insight into which workloads dynamic resizing is valuable for.
To accomplish this, we provide a mixture of analytic results and trace-driven numerical simulations in this section.

It is important to note that the case studies that follow depend fundamentally on the modeling performed so far in the paper, which allows us to capture and adjust independently, both fast time-scale and slow time-scale properties of the workload. The generality of our model framework enables thus a rigorous study of the impact of the workload on value of dynamic resizing.

\subsection{Setup}

Throughout the experimental setup, our aim is to choose parameters that provide conservative estimates of the case savings from dynamic resizing.  Thus, one should interpret the savings shown as a lower-bound on the potential savings.

\subsubsection*{Model Parameters}
The time frame for adapting the number of servers $n_k$ is assumed
to be $10$ min, and each time slot is assumed to be $1$ s, i.e.,
$U=600$.  When not otherwise specified, we assume the following
parameters for the data center $SLA$ agreement: the (virtual)
delay upper bound $\bar{D}=200$ms, and the delay violation
probability $\bar{\varepsilon}=10^{-3}$.

The cost is characterized by the two parameters of $e_{0}$ and
$e_{1}$, and the switching cost $\beta$. We choose units such that
the fixed energy cost is $e_{0}=1$. The load-dependent energy
consumption is set to $e_{1}=0$, because the energy consumption of
current servers with typical utilization level is dominated by the fixed costs \cite{spec,BH07,Kusic08}.  Note that adjusting $e_0$ and $e_1$ changes the magnitude of potential savings under dynamic resizing, but does not affect the qualitative conclusions about the impact of the workload.  So, due to space constraints, we fix these parameters during the case studies.

The normalized switching cost $\beta / e_{0}$ measures the
duration a server must be powered down to outweigh the switching
cost. Unless otherwise specified, we use $\beta = 6$, which corresponds to the energy
consumption for one hour (six frames). This was chosen as an
estimate of the time a server should sleep so that the
wear-and-tear of power cycling matches that of operating \cite{BACFJP08,LinWAT11}.

\subsubsection*{Workload Information}

\begin{figure*}[t]
\begin{center}
\subfigure[Hotmail]
{\includegraphics[width=0.49\linewidth]{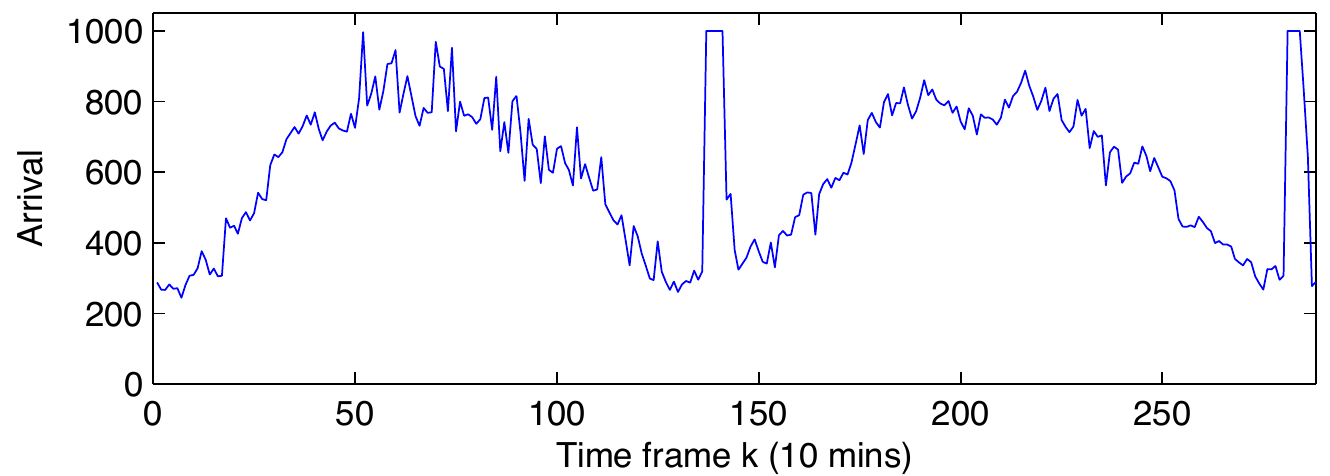}}
\subfigure[MSR]
{\includegraphics[width=0.49\linewidth]{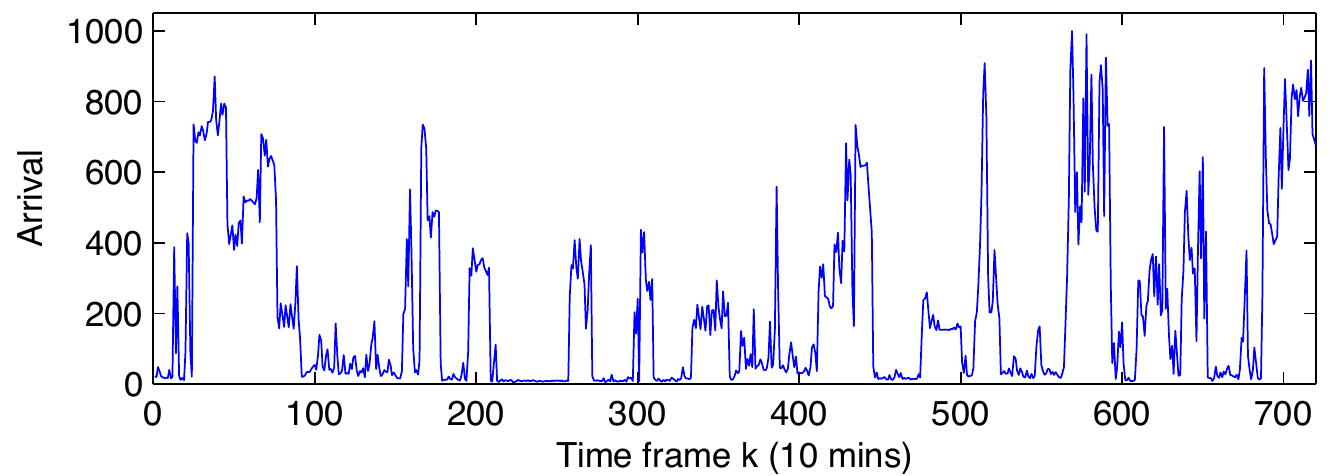}}
\caption{Illustration of the traces used for numerical experiments.} \label{fig:traces}
\end{center}
\end{figure*}

The workloads for these experiments are drawn from two real-world
data center traces. The first set of traces is from Hotmail, a
large email service running on tens of thousands of servers. We
used traces from 8 such servers over a 48-hour period, starting at
midnight (PDT) on Monday August 4 2008 \cite{Thereska09}. The
second set of traces is taken from 6 RAID volumes at MSR
Cambridge. The traced period was 1 week starting from 5PM GMT on
the 22nd February 2007 \cite{Thereska09}. Thus, these activity
traces represent a service used by millions of users and a small
service used by hundreds of users. The traces are normalized as
peak load $\lambda_{peak}=1000$, and are visualized in
Figure~\ref{fig:traces}. Both sets of traces show strong diurnal
properties and have peak-to-mean ratios (PMRs) of $1.64$ and
$4.64$ for Hotmail and MSR respectively. Time is partitioned into
10-minute frames and the load is averaged over each frame.

%The Hotmail trace contains significant nightly activity due to maintenance processes (backup, index creation etc). The data center, however, is provisioned for the peak foreground activity. This creates a dilemma: should our experiments include the maintenance activity or to remove it? If the spike is retained, it makes up nearly $12\%$ of the total load and forces the static provisioning to use a much larger number of servers than if it were removed, making savings from dynamic right-sizing much more dramatic. To provide conservative estimates of the savings from right-sizing, we chose to trim the size of the spike to minimize the savings from right-sizing.

The traces provide information for the slow time-scale model.  To parameterize the fast time-scale model, we adapt the workload based on the mean arrival rate in each frame, i.e., $\lambda$. To parameterize the MM processes, we take $\lambda_{l}=0.5\lambda$, $\lambda_{h}=2\lambda$, and we adjust the burst parameter $T$ while keeping $\lambda$ fixed for each process. To parameterize the heavy-tailed processes, we adjust the tail index $\alpha$ for each process, and $b$ in (\ref{e.pareto}) is adapted accordingly in order to keep the mean fixed at $\lambda$. Unless otherwise stated, we fix $\alpha=1.5$ and $T=1$.

\subsubsection*{Comparative Benchmark}

We contrast three designs: (i) the optimal dynamic resizing, (ii) dynamic resizing via LCP, and (iii) the optimal `static' provisioning.

The results for the optimal dynamic resizing should be interpreted as characterizing the potential of dynamic resizing.  But, realizing this potential is a challenge that requires both sophisticated online algorithms and excellent predictions of future workloads.\footnote{Note that short-term predictions of workload demand within 24 hours can be quite accurate~\cite{Gmach07,Kusic08}.}

The results for LCP should be interpreted as one example of how much of the potential for dynamic resizing can be attained with an online algorithm.  One reason for choosing LCP is that it does not rely on predicting the workload in future frames, and thus provides a conservative bound on the achievable cost savings.

\begin{figure*}[t]
\begin{center}
\subfigure[Hotmail]
{\includegraphics[width=0.49\linewidth]{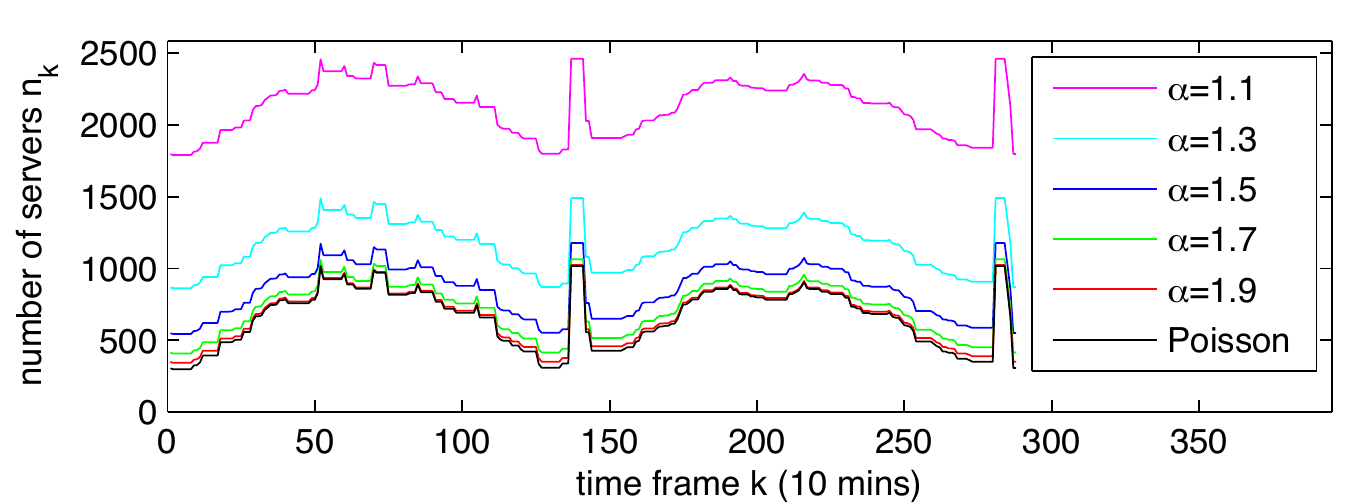}}
\subfigure[MSR]
{\includegraphics[width=0.49\linewidth]{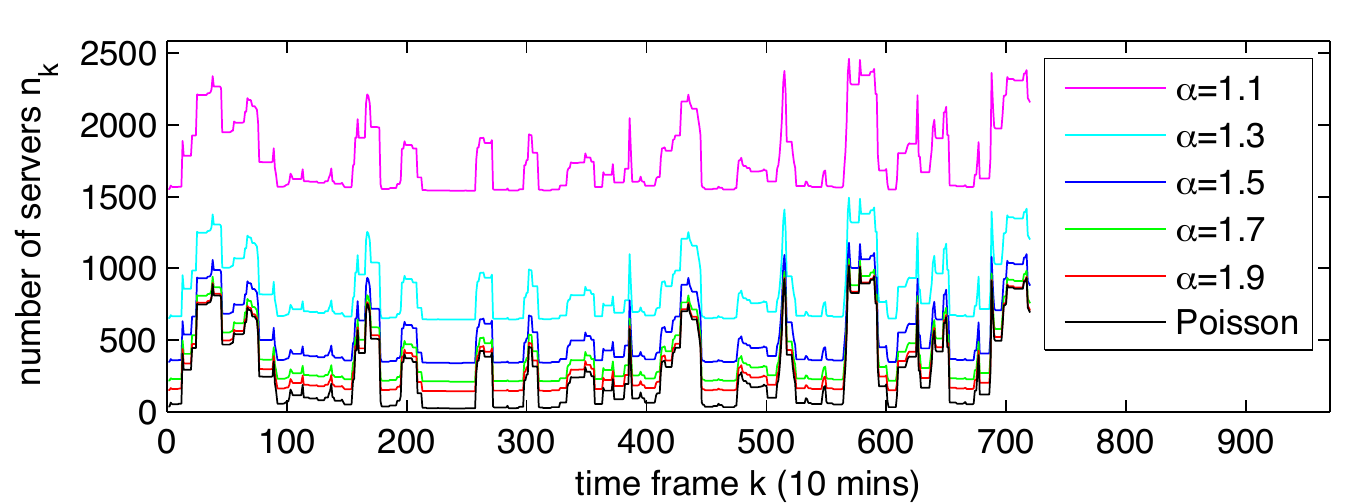}}
\caption{Impact of burstiness on provisioning $n_{k}$ for heavy-tailed arrivals.}
\label{fig:nkvaryalpha}
\subfigure[Hotmail]
{\includegraphics[width=0.49\linewidth]{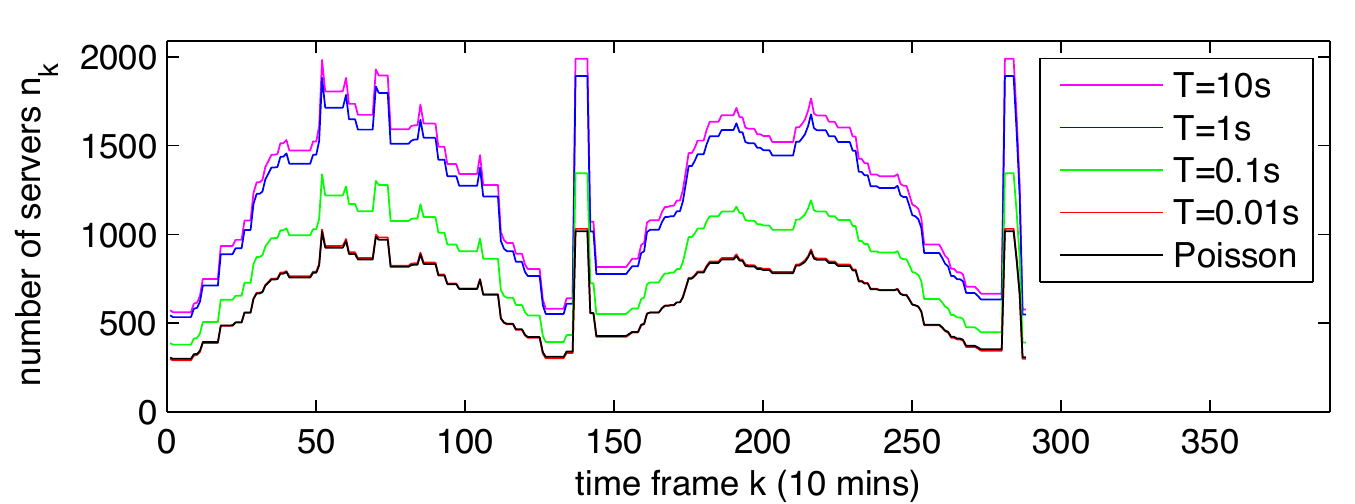}}
\subfigure[MSR]
{\includegraphics[width=0.49\linewidth]{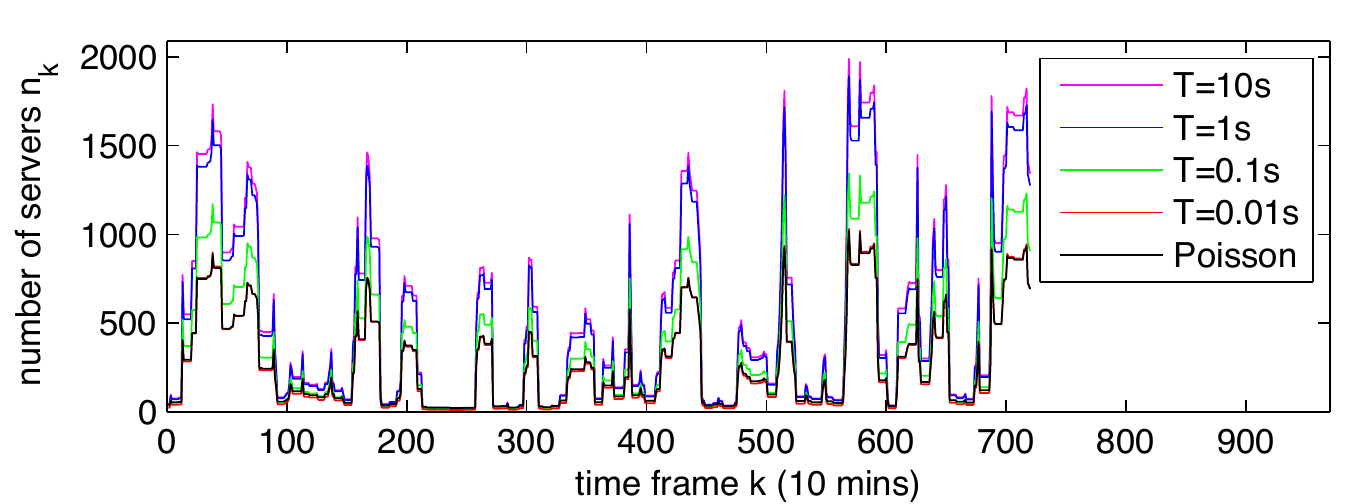}}
\caption{Impact of burstiness on provisioning $n_{k}$ for MM arrivals.}
\label{fig:nkvaryT}
\end{center}
\end{figure*}

\begin{figure}[t]
\begin{center}
\subfigure[Hotmail]
{\includegraphics[width=0.35\linewidth]{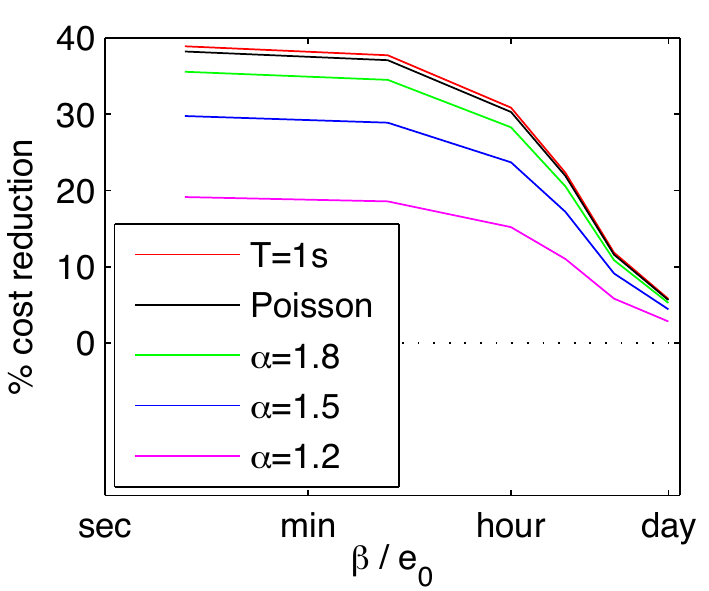}}
\subfigure[MSR]
{\includegraphics[width=0.35\linewidth]{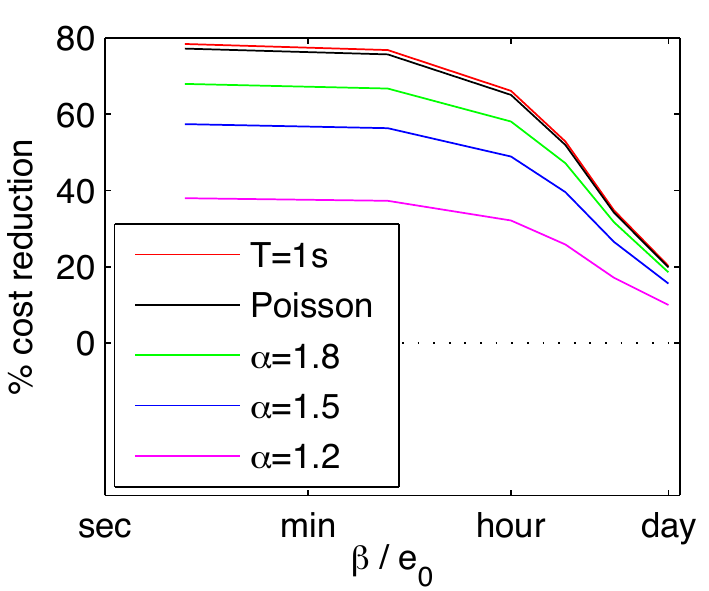}}
\caption{Impact of burstiness on the cost savings of dynamic resizing for different switching costs, $\beta$.}
\label{fig:savingsvaryburstiness}
\end{center}
\end{figure}

\begin{figure}[t]
\begin{center}
\subfigure[cost savings]
{\includegraphics[width=0.35\linewidth]{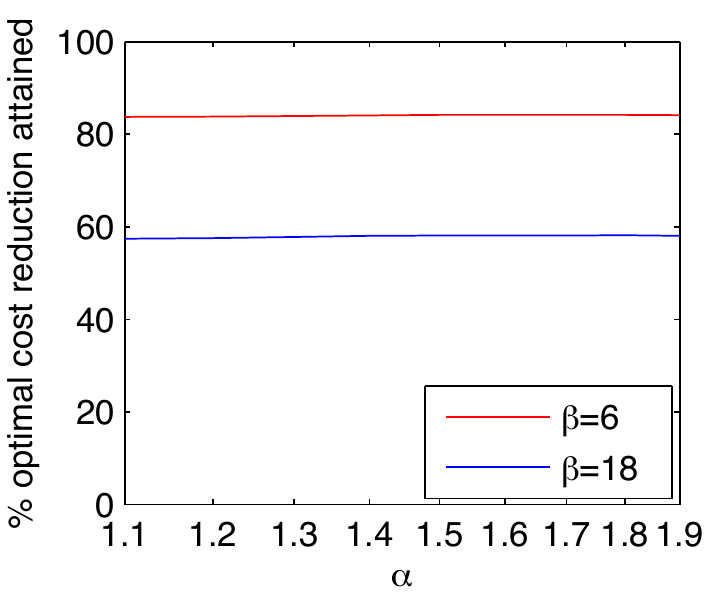}}
\subfigure[relative savings of LCP]
{\includegraphics[width=0.35\linewidth]{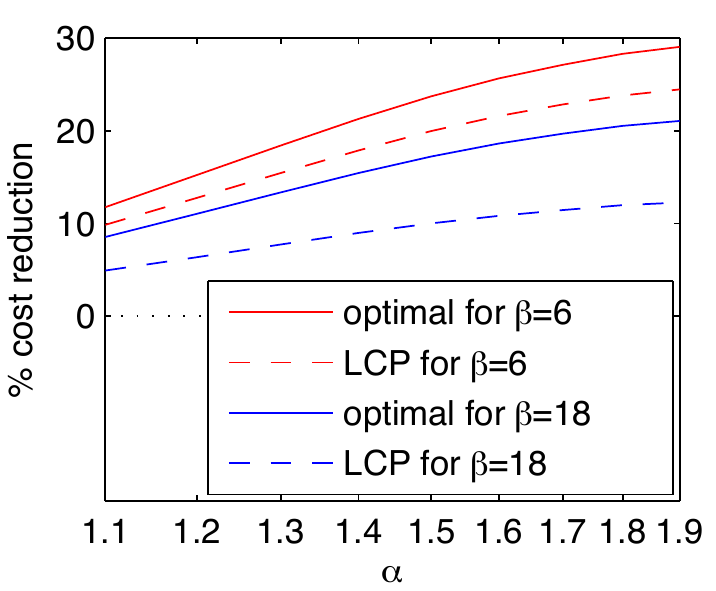}}
\caption{Impact of burstiness on the performance of LCP in the Hotmail trace.}
\label{fig:LCPsvaryburstiness}
\end{center}
\end{figure}

\begin{figure}[t]
\begin{center}
\subfigure[Hotmail]
{\includegraphics[width=0.35\linewidth]{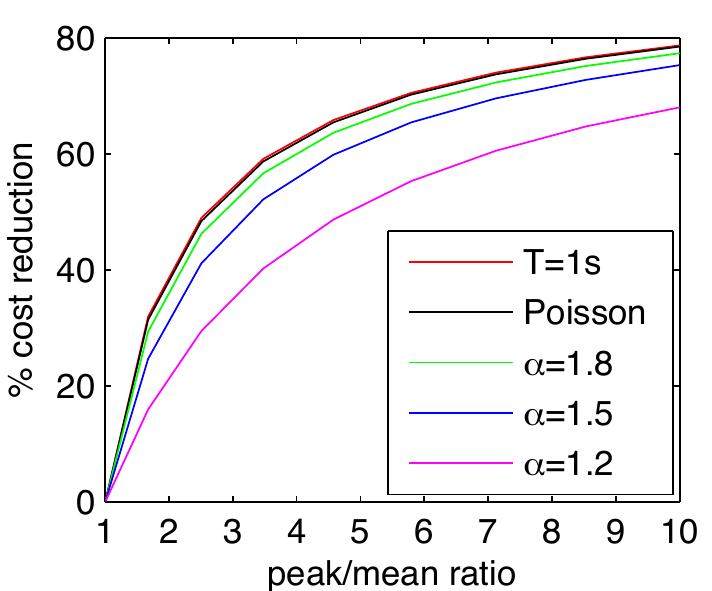}}
\subfigure[MSR]
{\includegraphics[width=0.35\linewidth]{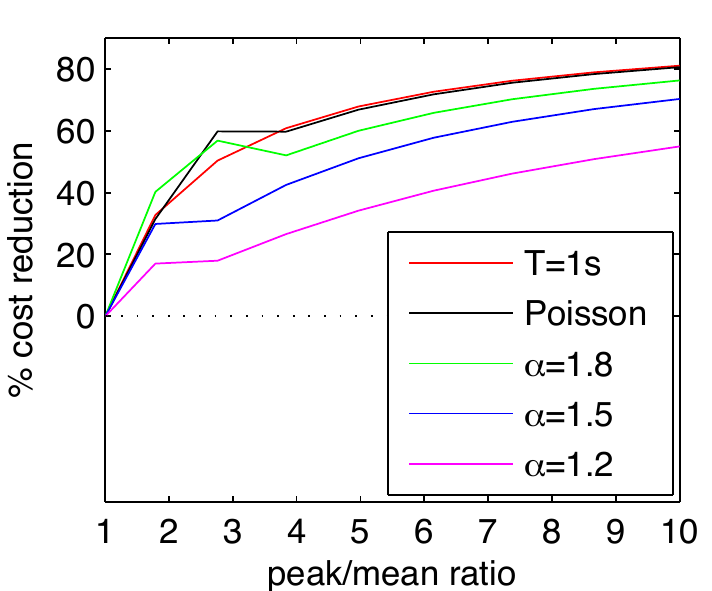}}
\caption{Impact of peak-to-mean ratio on the cost savings of the optimal dynamic resizing.}
\label{fig:savingsvarypeaktomean}
\end{center}
\end{figure}

\begin{figure*}[t]
\begin{center}
\subfigure[Hotmail]
{\includegraphics[width=0.49\linewidth]{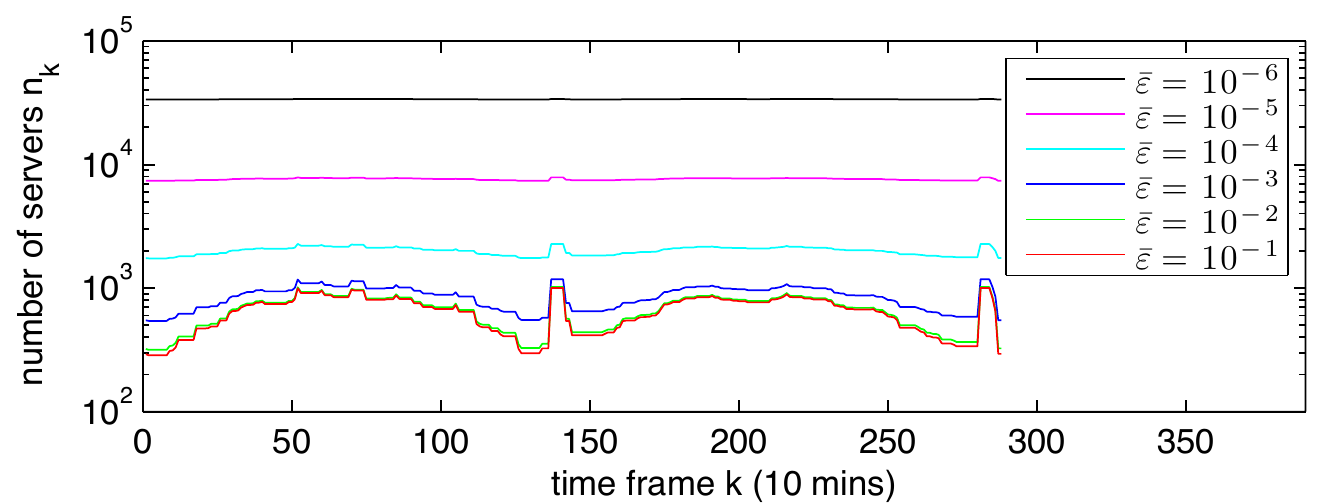}}
\subfigure[MSR]
{\includegraphics[width=0.49\linewidth]{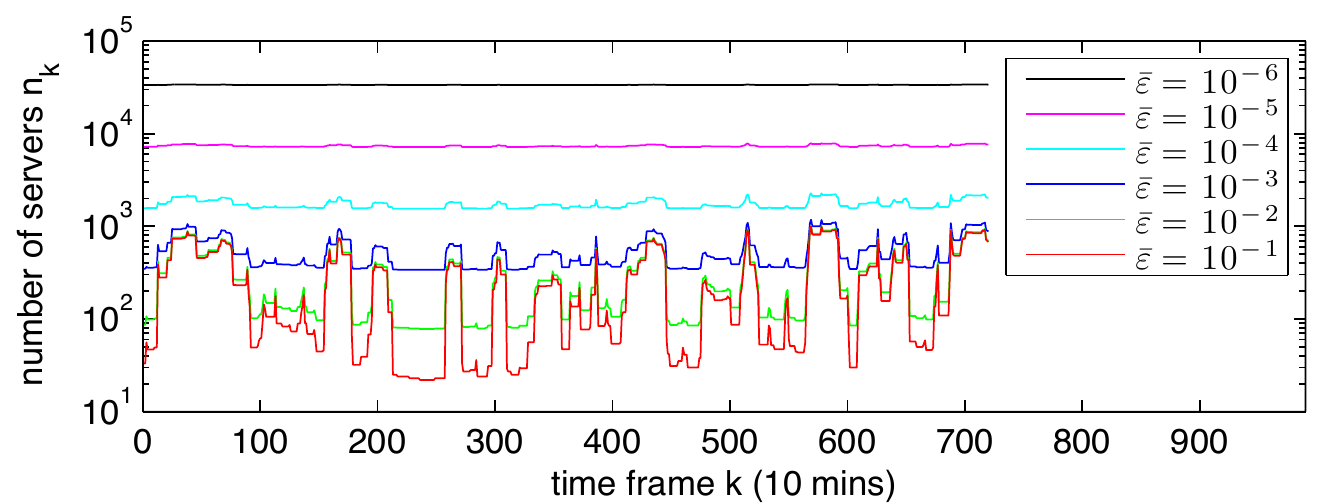}}
\caption{Impact of $\bar{\varepsilon}$ on provisioning $n_k$ for heavy tailed arrivals.}
\label{fig:nkvaryepsilonHT}
\subfigure[Hotmail]
{\includegraphics[width=0.49\linewidth]{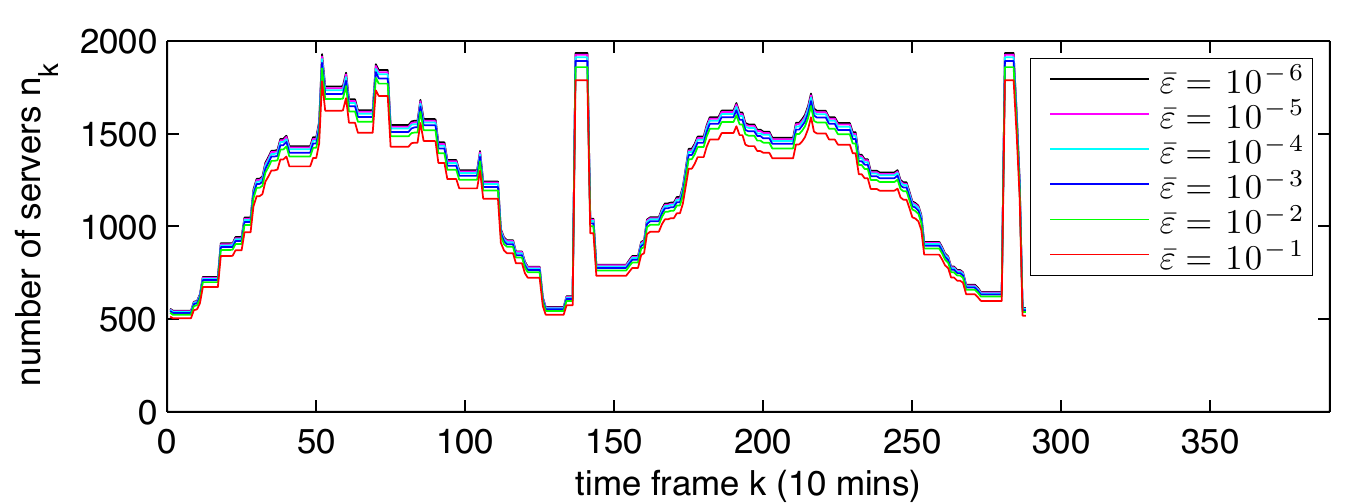}}
\subfigure[MSR]
{\includegraphics[width=0.49\linewidth]{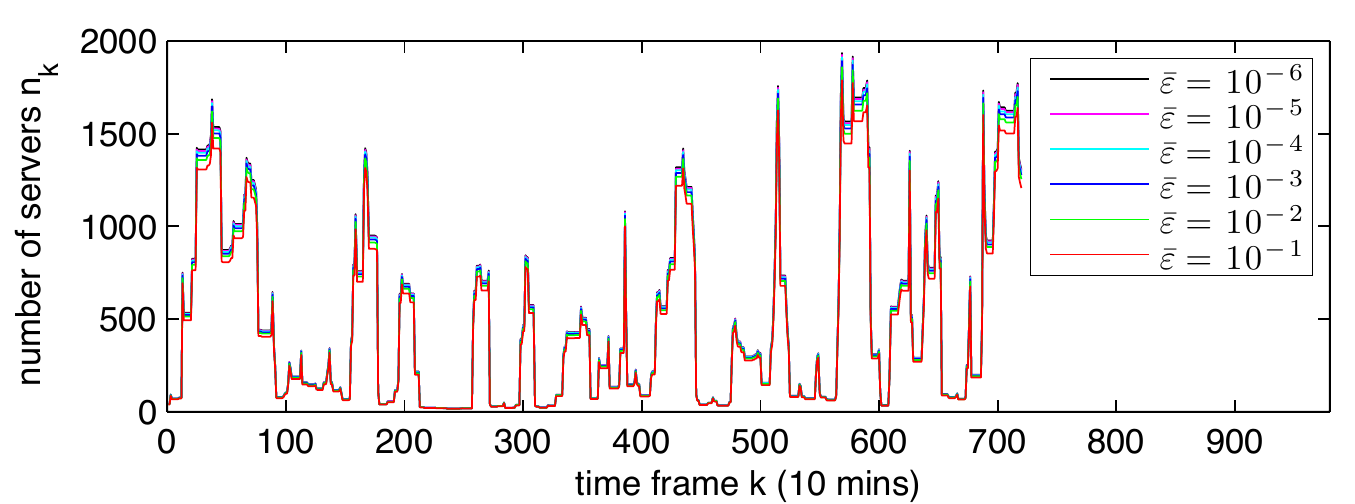}}
\caption{Impact of $\bar{\varepsilon}$ on provisioning $n_k$ for MM arrivals.}
\label{fig:nkvaryepsilonMM}

\subfigure[Hotmail]
{{\includegraphics[width=0.49\linewidth]{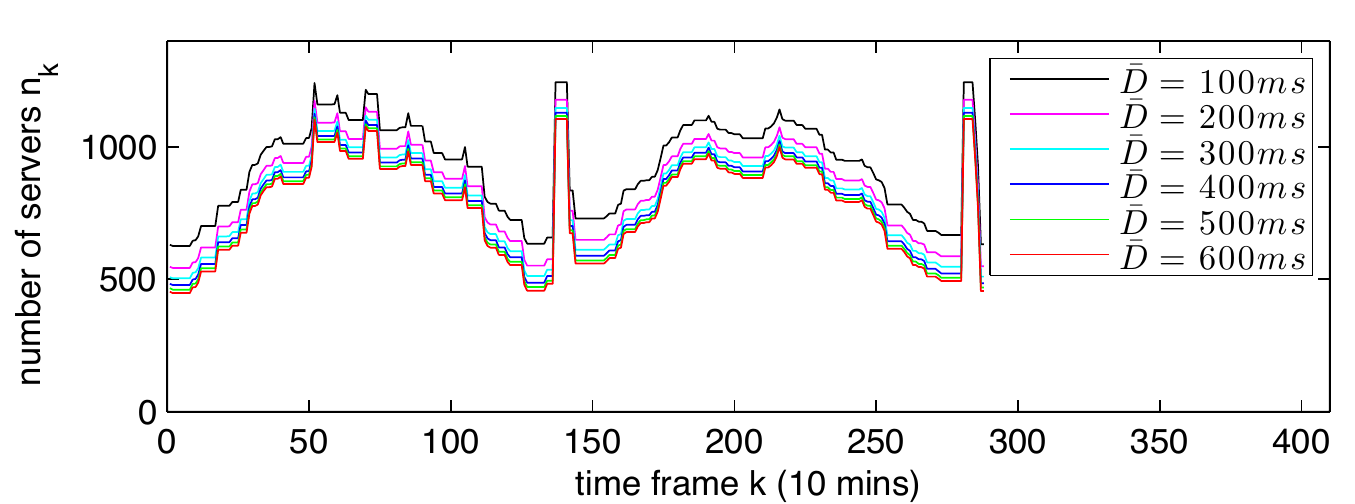}}}
\subfigure[MSR]
{{\includegraphics[width=0.49\linewidth]{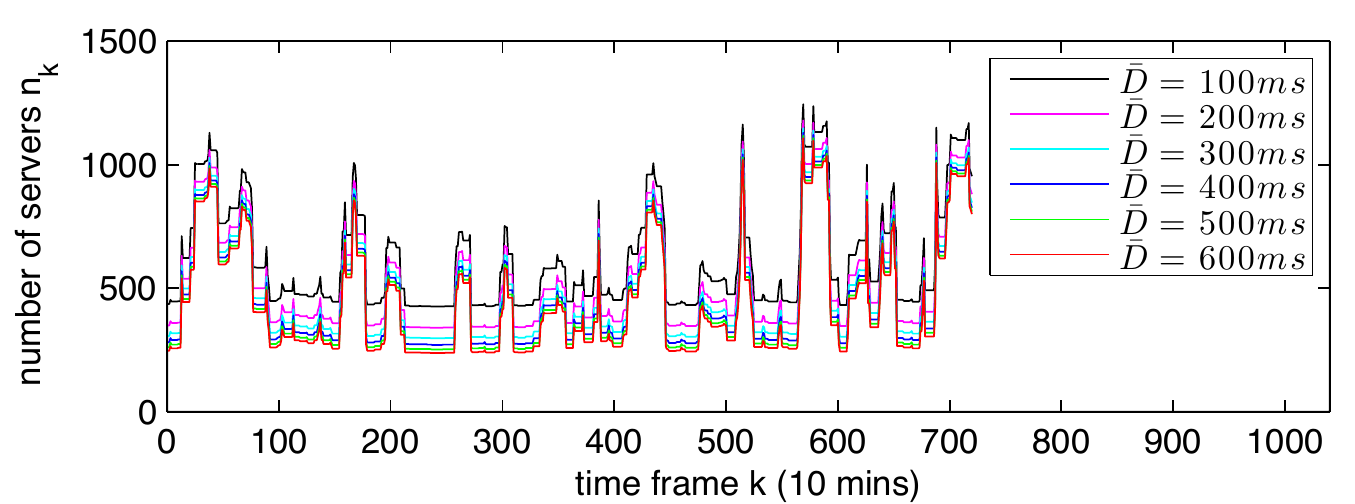}}}
\caption{Impact of $\bar{D}$ on provisioning $n_k$ for heavy tailed arrivals.}
\label{fig:nkvaryDHT}
\end{center}
\end{figure*}

The results for the optimal static provisioning should be taken as an optimistic benchmark for today's data centers, which typically do not use dynamic resizing. We consider the cost incurred by an  optimal static provisioning  scheme that chooses a constant number of servers that minimizes the costs incurred based on full knowledge of the entire workload. This
policy is clearly not possible in practice, but it provides a \emph{very conservative} estimate of the savings from right-sizing since it uses perfect knowledge of all peaks and eliminates the need for overprovisioning in order to handle the possibility of flash crowds or other traffic bursts.

\subsection{Results}

Our experiments are organized to illustrate the impact of a wide variety of parameters on the cost savings attainable via dynamic resizing.  The goal is to understand for which workloads dynamic resizing can provide large enough cost savings to warrant the extra implementation complexity.  Remember, our setup is designed so that the cost savings illustrated is a conservative estimate of the true cost savings provided by dynamic resizing.

\subsubsection*{The Role of Burstiness}

A key goal of our model is to expose the impact of burstiness on
dynamic resizing, and so we start by focusing on that parameter.
Recall that we can vary burstiness in both the light-tailed and
heavy-tailed settings using $T$ for MM arrivals and $\alpha$ for
heavy-tailed arrivals.

\textbf{The impact of burstiness on provisioning:} A priori, one may expect that burstiness can be beneficial for dynamic resizing, since it indicates that there are periods of low load during which energy may be saved.  However, this is not actually true since resizing decisions must be made at the slow time-scale while burstiness is a characteristic of the fast time-scale.  Thus, burstiness is actually detrimental for dynamic resizing, since it means that the provisioning decisions made on the slow time-scale must be made with the bursts in mind, which results in a larger number of servers needed to be provisioned for the same average workload. This effect can be seen in Figures \ref{fig:nkvaryalpha} and \ref{fig:nkvaryT}, which show the optimal dynamic provisioning as $\alpha$ and $T$ vary.  Recall that burstiness increases as $\alpha$ decreases and $T$ increases.

\textbf{The impact of burstiness on cost savings:}  The larger provisioning created by increased burstiness manifests itself in the cost savings attainable through dynamic capacity provisioning as well.  This is illustrated in Figure \ref{fig:savingsvaryburstiness}, which shows the cost savings of the optimal dynamic provisioning as compared to the optimal static provisioning for varying $\alpha$ and $T$ as a function of the switching cost $\beta$.

\textbf{The impact of burstiness on LCP:} Interestingly, though Figure \ref{fig:savingsvaryburstiness} shows that the potential of dynamic resizing is limited by increased burstiness, it turns out that the relative performance of LCP is not hurt by burstiness.  This is illustrated in Figure \ref{fig:LCPsvaryburstiness}, which shows the percent of the optimal cost savings that LCP achieves. Importantly, it is nearly perfectly flat as the burstiness is varied.

\subsubsection*{The Role of the Peak-to-Mean Ratio}

The impact of the peak-to-mean ratio on the potential benefits of dynamic resizing is quite intuitive:  if the peak-to-mean ratio is high, then there is more opportunity to benefit from dynamically changing capacity.  Figure \ref{fig:savingsvarypeaktomean} illustrates this well-known effect. The workload for the figure is
generated from the traces by scaling $\lambda_{k}$ as
$\hat{\lambda}_{k}=c(\lambda_{k})^{\gamma}$, varying $\gamma$ and
adjusting $c$ to keep the mean constant.

In addition to illustrating that a higher peak-to-mean ratio makes dynamic resizing more valuable, Figure \ref{fig:savingsvarypeaktomean} also highlights that there is a strong interaction between burstiness and the peak-to-mean ratio, where if there is significant burstiness the benefits that come from a high peak-to-mean ratio may be diminished considerably.

\subsubsection*{The Role of the SLA}

The SLA plays a key role in the provisioning of a data center.  Here, we show that the SLA can also have a strong impact on whether dynamic resizing is valuable, and that this impact depends on the workload.  Recall that in our model the SLA consists of a violation probability $\bar{\varepsilon}$ and a delay bound $\bar{D}$.  We deal with each of these in turn.

Figures \ref{fig:nkvaryepsilonHT} and \ref{fig:nkvaryepsilonMM} highlight the role the violation probability $\bar{\varepsilon}$ has on the provisioning of $n_k$ under the optimal dynamic resizing in the cases of heavy-tailed and MM arrivals.  Interestingly, we see that there is a significant difference in the impact of $\bar{\varepsilon}$ depending on the arrival process.  As $\bar{\varepsilon}$ gets smaller in the heavy-tailed case the provisioning gets significantly flatter, until there is almost no change in $n_k$ over time.  In contrast, no such behavior occurs in the MM case and, in fact, the impact of $\bar{\varepsilon}$ is quite small.  This difference is a fundamental effect of the ``heaviness'' of the tail of the arrivals, i.e., a heavy tail requires significantly more capacity in order to counter a drop in $\bar{\varepsilon}$.

This contrast between heavy- and light-tailed arrivals is also evident in Figure \ref{fig:savingsvaryepsilon}, which highlights the cost savings from dynamic resizing in each case as a function of $\bar{\varepsilon}$.  Interestingly, the cost savings under light-tailed arrivals is largely independent of $\bar{\varepsilon}$, while under heavy-tailed arrivals the cost savings is monotonically increasing with $\bar{\varepsilon}$.

The second component of the SLA is the delay bound $\bar{D}$. The impact of $\bar{D}$ on provisioning is much less dramatic.  We show an example in the case of heavy-tailed arrivals in Figure \ref{fig:nkvaryDHT}.  Not surprisingly, the provisioning increases as $\bar{D}$ drops.  However, the flattening observed as a result of $\bar{\varepsilon}$ is not observed here.  The case of MM arrivals is qualitatively the same, and so we do not include it.

\subsubsection*{When is Dynamic Resizing Valuable?}

Now, we are finally ready to address the question of when (i.e., for what workloads) is dynamic resizing valuable.
To address this question, we must look at the interaction between the peak-to-mean ratio and the burstiness.  Our goal is to provide a concrete understanding of for which (peak-to-mean, burstiness, SLA) settings the potential savings from dynamic resizing is large enough to warrant implementation.  Figures \ref{fig:whenvaluablebysavings}--\ref{fig:whenvaluablebyepsilon} focus on this question.  Our hope is that these figures highlight that a precursor to any debate about the value of dynamic resizing must be a joint understanding of the expected workload characteristics and the desired SLA, since for any fixed choices of two of these parameters (peak-to-mean, burstiness, SLA), the third can be chosen so that dynamic resizing does or does not provide significant cost savings for the data center.

Starting with Figure \ref{fig:whenvaluablebysavings}, we see a set of curves for different levels of cost savings.  The interpretation of the figures is that below (above) each curve the savings from optimal dynamic resizing is smaller (larger) than the specified value for the curve.  Thus, for example, if the peak-to-mean ratio is 2 in the Hotmail trace, a 10\% cost savings is possible for all levels of burstiness, but a 30\% cost savings is only possible for $\alpha>1.5$.  However, if the peak-to-mean ratio is 3, then a 30\% cost savings is possible for all levels of burstiness.  It is difficult to say what peak-to-mean and burstiness settings are ``common'' for data centers, but as a point of reference, one might expect large-scale services to have a peak-to-mean ratio similar to that of the Hotmail trace, i.e., around 1.5-2.5; and smaller scale services to have peak-to-mean ratios similar to that of the MSR trace, i.e., around 4-6.  The burstiness also can vary widely, but as a rough estimate, one might expect $\alpha$ to be around 1.4-1.6.

\begin{figure}[t]
\begin{center}
\subfigure[Hotmail]
{\includegraphics[width=0.35\linewidth]{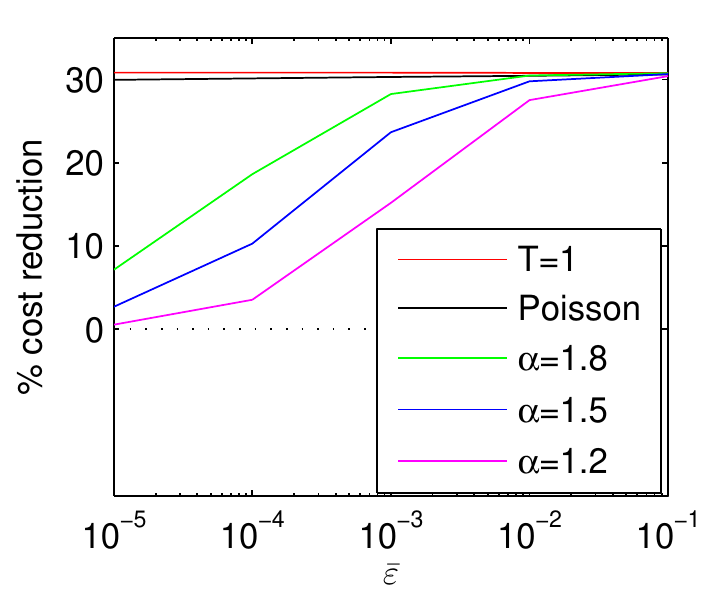}}
\subfigure[MSR]
{\includegraphics[width=0.35\linewidth]{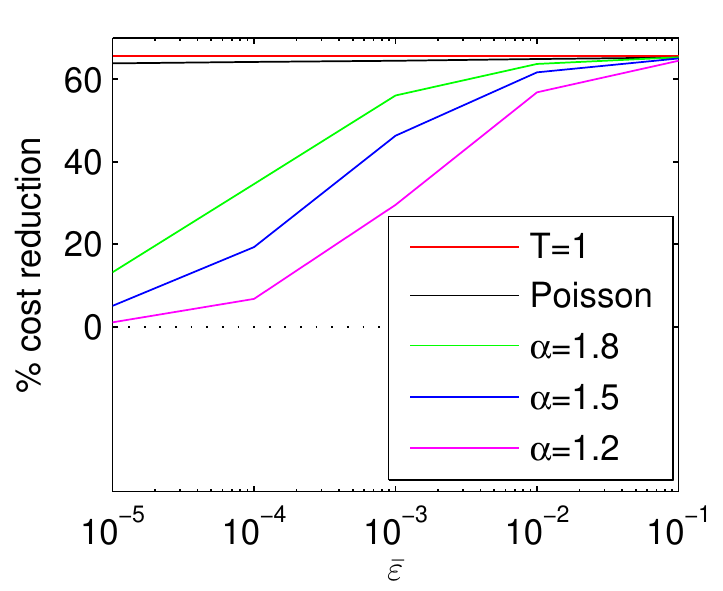}}
\caption{Impact of $\bar{\varepsilon}$ on the cost savings of dynamic resizing.}
\label{fig:savingsvaryepsilon}
\end{center}
\end{figure}

\begin{figure}[t]
\begin{center}
\subfigure[Hotmail]
{\includegraphics[width=0.35\linewidth]{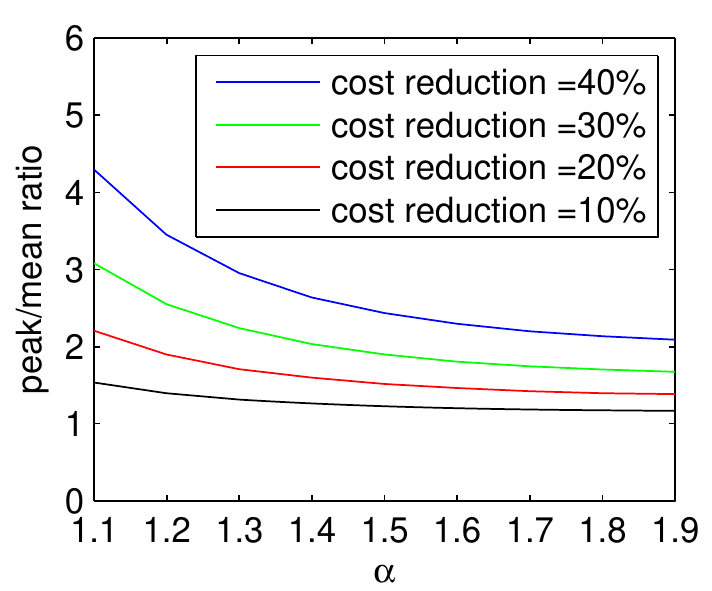}}
\subfigure[MSR]
{\includegraphics[width=0.35\linewidth]{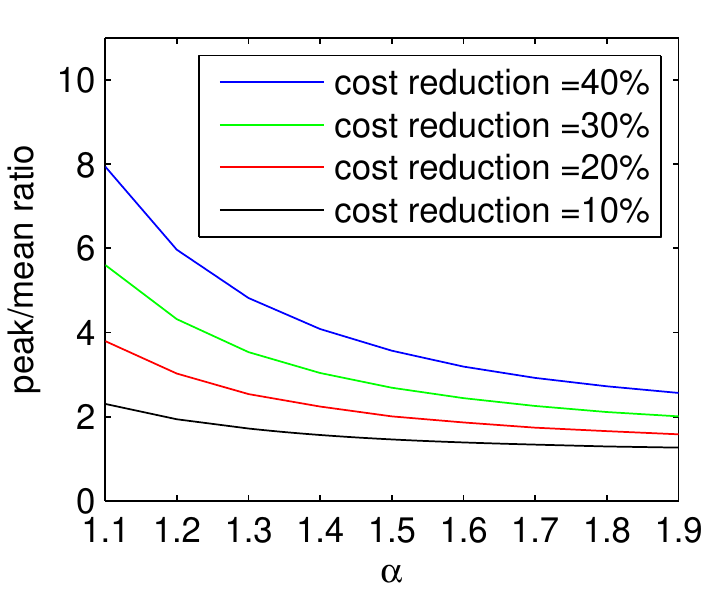}}
\caption{Characterization of burstiness and peak-to-mean ratio necessary for dynamic resizing to achieve different levels of cost reduction.}
\label{fig:whenvaluablebysavings}
\end{center}
\end{figure}

Of course, many of the settings of the data center will effect the conclusions illustrated in Figure \ref{fig:whenvaluablebysavings}.  Two of the most important factors to understand the effects of are the switching cost, $\beta$, and the SLA, particularly $\bar{\varepsilon}$.

Figure \ref{fig:whenvaluablebybeta} highlights the impact of the magnitude of the switching costs on the value of dynamic resizing.  The curves represent the threshold on peak-to-mean ratio and burstiness necessary to obtain 20\% cost savings from dynamic resizing.  As the switching costs increase, the workload must have a larger peak-to-mean ratio and/or less burstiness in order for dynamic resizing to be valuable.  This is not unexpected.  However, what is perhaps surprising is the small impact played by the switching cost.  The class of workloads where dynamic resizing is valuable only shrinks slightly as the switching cost is varied from on the order of the cost of running a server for 10 minutes ($\beta=1$) to running a server for 3 hours ($\beta=18$).

\begin{figure}[t]
\begin{center}
\subfigure[Hotmail]
{\includegraphics[width=0.35\linewidth]{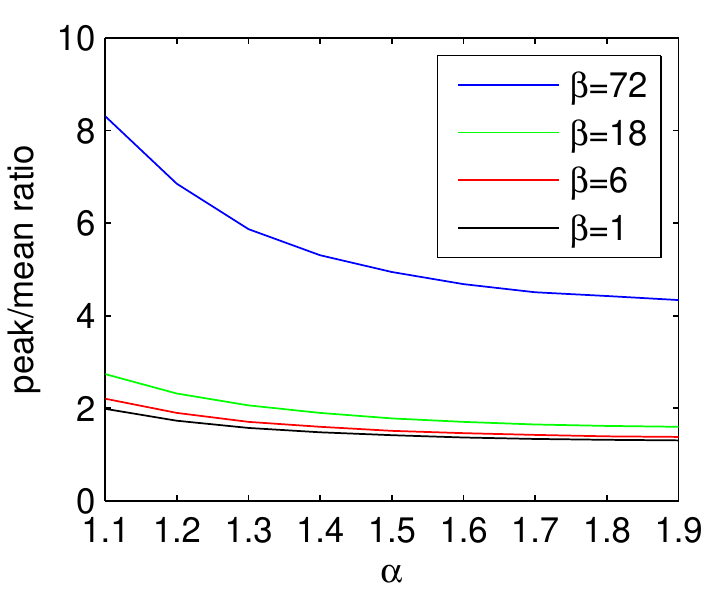}}
\subfigure[MSR]
{\includegraphics[width=0.35\linewidth]{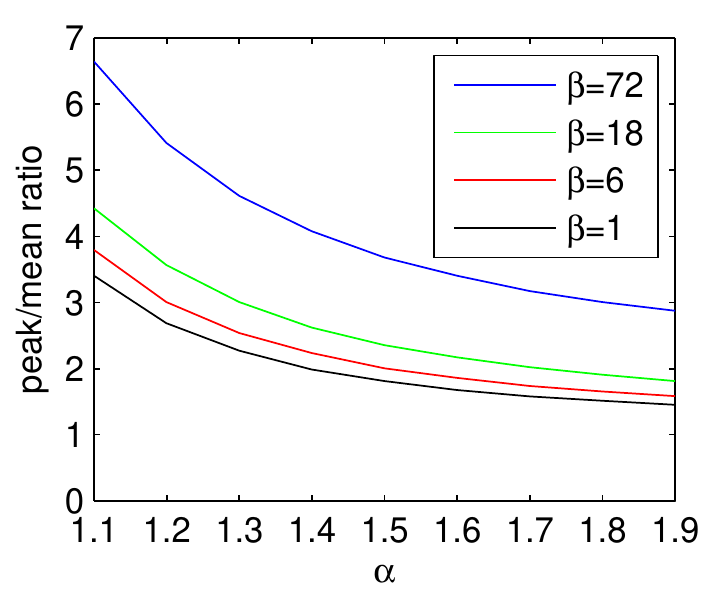}}
\caption{Characterization of burstiness and peak-to-mean ratio necessary for dynamic resizing to achieve 20\% cost reduction as a function of the switching cost, $\beta$.}
\label{fig:whenvaluablebybeta}
\end{center}
\end{figure}

Interestingly, while the impact of the switching costs on the value of dynamic resizing is small, the impact of the SLA is quite large.  In particular, the violation probability $\bar{\varepsilon}$ can dramatically affect whether dynamic resizing is valuable or not.  This is shown in Figure \ref{fig:whenvaluablebyepsilon}, on which the curves represent the threshold on peak-to-mean ratio and burstiness necessary to obtain 20\% cost savings from dynamic resizing.  We see that, as the violation probability is allowed to be larger, the impact of the peak-to-mean ratio on the potential of savings from dynamic resizing disappears; and the value of dynamic resizing starts to depend almost entirely on the burstiness of the arrival process.  The reason for this can be observed in Figure \ref{fig:nkvaryepsilonHT}, which highlights that the optimal provisioning $n_k$ becomes nearly flat as $\bar{\varepsilon}$ increases.

\begin{figure}[t]
\begin{center}
\subfigure[Hotmail]
{\includegraphics[width=0.35\linewidth]{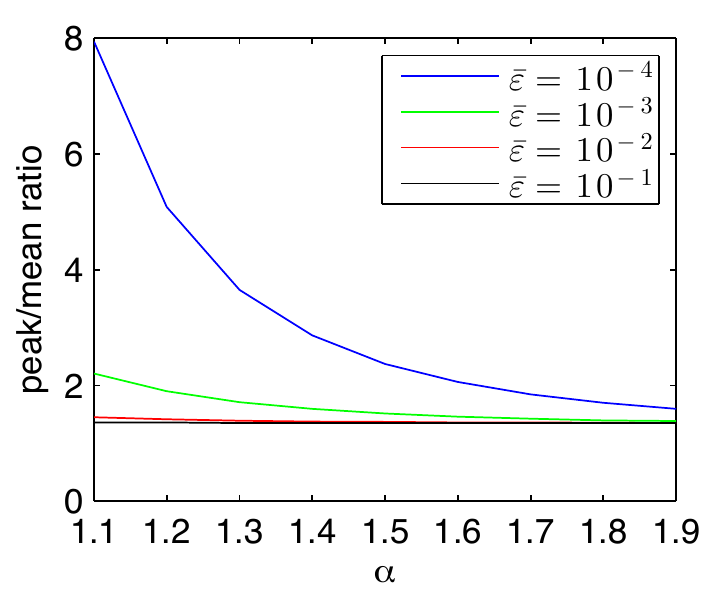}}
\subfigure[MSR]
{\includegraphics[width=0.35\linewidth]{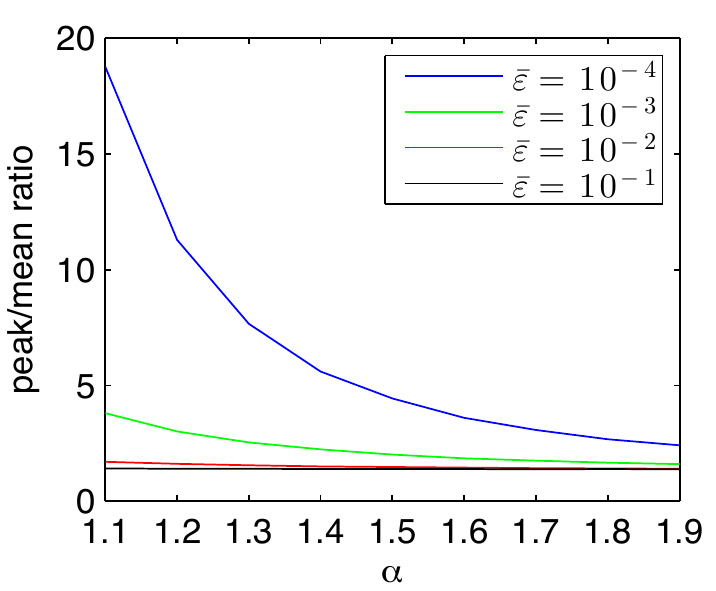}}
\caption{Characterization of burstiness and peak-to-mean ratio necessary for dynamic resizing to achieve 20\% cost reduction as a function of the SLA, $\bar{\varepsilon}$.}
\label{fig:whenvaluablebyepsilon}
\end{center}
\end{figure}

\subsubsection*{Supporting Analytic Results}

To this point we have focused on numerical simulations, and further we provide analytic support for the behavior we observed in the experiments above.  In particular, the following two theorems characterize the impact of burstiness and the SLA $(\bar{D},\bar{\eps})$ on the value of dynamic resizing under Poisson and heavy-tailed arrivals.  This is accomplished by deriving the effect of these parameters on $C(\bar{D},\bar{\eps})$, which constrains the optimal provisioning $n_k$.  A smaller (larger) $C(\bar{D},\bar{\eps})$ implies a smaller (larger) provisioning $n_k$, which in turn implies smaller (larger) costs.

We start providing a result for the case of Poisson arrivals.  The proof is given in Section \ref{s.proofs}.

\begin{theorem} \label{t.ppscaling}
The service capacity constraint from Eq.~(\ref{eq:CPoisson})
increases as the delay constraint $\bar{D}$ or the violation
probability $\bar{\eps}$ decrease. It also satisfies the scaling
law
\begin{equation*}
C(\bar{D},\bar{\eps})=\Theta\left(\frac{\bar{D}^{-1}\log{\bar{\eps}^{-1}}}{\log{(\bar{D}^{-1}\log{\bar{\eps}^{-1}})}}\right),
\end{equation*}
as $\bar{D}^{-1}\log{\bar{\eps}^{-1}}\rightarrow\infty$.
\end{theorem}

This theorem highlights that as $\bar{\eps}$ decreases and/or $\bar{D}$ decreases $C(\bar{D},\bar{\eps})$, and thus the cost of the optimal provisioning, increases.  This shows that the observations made in our numeric experiments hold more generally.  Perhaps the most interesting point about this theorem, however, is the contrast of the growth rate with that in the case of heavy-tailed arrivals, which is summarized in the following theorem.  The proof is given in Section \ref{s.proofs}.

\begin{theorem} \label{t.htscaling}
The implicit solution for the capacity constraint from
Eq.~(\ref{eq:solHT}) increases as the delay constraint $\bar{D}$
or the violation probability $\bar{\eps}$ decrease, or the value
of $\alpha$ decreases. It also satisfies the scaling law
\begin{equation}
C(\bar{D},\bar{\eps})=\Theta\left(\left(\frac{1}{\bar{\eps}\bar{D}^{\alpha-1}}\right)^{\frac{1}{\alpha}}\right)\label{eq:htssLaws}
\end{equation}
as $\bar{\eps}\bar{D}^{\alpha-1}\rightarrow 0$
 for any given $\alpha\in(1,2)$.
\end{theorem}

A key observation about this theorem is that the growth rate of $C(\bar{D},\bar{\eps})$ with $\bar{\eps}$ is much faster than in the case of the Poisson (polynomial instead of logarithmic).  This supports what is observed in Figure \ref{fig:savingsvaryepsilon}. Additionally, Theorem \ref{t.htscaling} highlights the impact of burstiness, $\alpha$, and shows that the behavior we have seen in our experiments holds more generally.

%Note also that by letting $\bar{D}^{-1}\log{\bar{\eps}^{-1}}\rightarrow 0$ in Eq.~(\ref{eq:CPoisson}) (i.e., large delays $\bar{D}$ and/or large violation probabilities $\bar{\eps}$) we get the expected behavior $C(\bar{D},\bar{\eps})\rightarrow\lambda$.

%%%%%%%%%%%%%%%%%%%%%%%%%%%%%%%%%%%%%%%%%%%%%%

\section{Proofs}
\label{s.proofs}

In this section, we collect the proofs for the results in previous sections.

We start with the proof of Lemma \ref{aggregate}, the aggregation property used to model the multiserver system with a single service process.

\begin{proof}[Proof of Lemma \ref{aggregate}]
Fix $t\geq1$. Because each server $i$ has a constant rate capacity
$\mu$, it follows that the bivariate processes
$S_i(s,t)=\mu(t-s)$ are service processes for the individual
servers~(see~\cite{Book-Chang}, pp. 167), i.e.,
\begin{eqnarray*}
R_i(t)&\geq&\inf_{0\leq s\leq t}\left\{A_i(s)+\mu(t-s)\right\}\\
&=&\frac{1}{n}\inf_{0\leq s\leq t}\left\{A(s)+n\mu(t-s)\right\}~,
\end{eqnarray*}
where $R_i(t)$ is the departure process from server $i$. In the
last line we used the load-balancing dispatching assumption, i.e.,
$A_i(s)=\frac{1}{n}A(s)$. Adding the terms for $i=1,\dots,n$ it
immediately follows that
\begin{center}
$\sum_{i} R_i(t)\geq\inf_{0\leq s\leq t}\left\{A(s)+n\mu(t-s)\right\}$,
\end{center}
which shows that the bivariate process $S(s,t)=n\mu(t-s)$ is a
service process for the virtual system with arrival process
$A(t)=\sum_i A_i(t)$ and departure process $R(t)=\sum_i R_i(t)$.
\end{proof}

Next, we prove the bound used for Poisson arrival processes, i.e.,
Eq. (\ref{eq:dbMM}).

\begin{proof}[Proof of Proposition~\ref{t.Poisson}]
The proof follows closely a technique from~\cite{Kingman64} for
the analysis of GI/GI/1 queues. Denote for convenience
$C=C(\bar{D},\bar{\eps})$. Fix $t\geq 1$ and introduce the
following process for all $0\leq s\leq t-\bar{D}$:
\begin{eqnarray*}
T(s)=e^{\theta^*\left(A(t-\bar{D}-s,t-\bar{D})-Cs\right)}.
\end{eqnarray*}
Consider also the filtration of
$\sigma$-algebras
\begin{eqnarray*}
{\mathcal{F}}_s=\sigma\left\{A(t-\bar{D}-s,t-\bar{D})\right\},
\end{eqnarray*}
i.e., ${\mathcal{F}}_s\subseteq{\mathcal{F}}_{s+u}$ for all $0\leq
s\leq s+u\leq t-\bar{D}$. Note that $T(s)$ is
${\mathcal{F}}_s$-measurable for all $s$ (see~\cite{Book-Revuz99},
pp.~79). Then we can write for the conditional expectations for
all $s,u\geq0$ with $s+u\leq t-\bar{D}$
\begin{eqnarray*}
&&\hspace{-0.7cm}E\left[T(s+u)\parallel{\mathcal{F}}_s\right]\\
&&\qquad=E\left[T(s)e^{\theta^*\left(A(t-\bar{D}-s-u,t-\bar{D}-s)-Cu\right)}\parallel{\mathcal{F}}_s\right]\\
&&\qquad=T(s)E\left[e^{\theta^*\left(A(t-\bar{D}-s-u,t-\bar{D}-s)-Cu\right)}\parallel{\mathcal{F}}_s\right]\\
&&\qquad=T(s)E\left[e^{\theta^*\left(A(t-\bar{D}-s-u,t-\bar{D}-s)-Cu\right)}\right]\\
&&\qquad=T(s)e^{\theta^*\left(\frac{\lambda}{\theta^*}\left(e^{\theta^*}-1\right)-C\right)u}\\
&&\qquad\leq T(s)~.
\end{eqnarray*}
In the second line we used that $T(s)$ is
${\mathcal{F}}_s$-measurable, and then we used the independent
increments property of the Poisson process $A(t)$, i.e.,
$A(t-\bar{D}-s-u,t-\bar{D}-s)$ is independent of
${\mathcal{F}}_s$. Then we computed the moment generating function
for the Poisson process, and finally we used the property of
$\theta^*$ from Eq.~(\ref{eq:thetaStar}). Therefore, the process
$T(s)$ is a supermartingale, i.e.,
\begin{eqnarray*}
E\left[T(s+u)\parallel{\mathcal{F}}_s\right]\leq T(s).
\end{eqnarray*}

We can now continue Eq.~(\ref{eq:db0}) as follows
\begin{eqnarray*}
&&\hspace{-0.5cm}\P\left(D(t)>\bar{D}\right)\\
&&\leq\P\left(\sup_{0\leq s\leq
t-\Bar{D}}\left\{A\left(s,t-\Bar{D}\right)-C(t-\bar{D}-s)\right\}>C\bar{D}\right)\\
&&\leq\P\left(\sup_{0\leq s\leq t-\Bar{D}}T(s)>e^{\theta^*C\bar{D}}\right)\\
&&\leq e^{-\theta^*C\bar{D}}~,
\end{eqnarray*}
which proves Eq.~(\ref{eq:dbMM}). Note that in the last line we
used a maximal inequality for the (continuous) supermartingale
$T(s)$ (see~\cite{Book-Revuz99}, pp.~54).
\end{proof}

Finally, we prove the monotonicity and scaling results in Theorems
\ref{t.ppscaling} and \ref{t.htscaling}.

\begin{proof}[Proof of Theorem \ref{t.ppscaling}]
First, note that the monotonicity properties follow immediately
from the fact that the function $f(x)=(1+x)^{\frac{1}{x}}$ is
non-decreasing. Next, to prove the more detailed scaling laws,
simply notice that $\log{(1+c
f(n))}=\Theta\left(\log{f(n)}\right)$ for some non-decreasing
function $f(n)$ and a constant $c>0$.  The result follows.
\end{proof}

\begin{proof}[Proof of Theorem \ref{t.htscaling}]

We first consider the monotonicity properties and then the scaling law.

\emph{Monotonicity properties:}  To prove the monotonicity results on $\bar{D}$ and $\bar{\eps}$,
observe that the left hand side (LHS) in the implicit equation from
Eq.~(\ref{eq:solHT}) is a non-increasing function in
$C(\bar{D},\bar{\eps})$ because the range of the infimum expands
whereas the function in the infimum decreases, by increasing
$C(\bar{D},\bar{\eps})$.  Moreover, the LHS is unbounded at the
boundary $C(\bar{D},\bar{\eps})=\lambda$. The solution
$C(\bar{D},\bar{\eps})$ is thus non-increasing in both $\bar{D}$
and $\bar{\eps}$.

Next, to prove monotonicity in $\alpha$, fix
$\alpha_1\leq\alpha_2$ and denote by $C_1$ the implicit solution
of Eq.~(\ref{eq:solHT}) for $\alpha=\alpha_1$. In the first step
we prove that $C_1\bar{D}\geq1$. Let $C$ be the solution of
\begin{equation*}
\frac{1}{C^{\alpha_1}}=\bar{\eps}\bar{D}^{\alpha_1-1},
\end{equation*}
where the LHS was obtained by relaxing the LHS of
Eq.~(\ref{eq:solHT}) (we used that $\gamma>1$,
$C-\gamma\lambda<C$, and $x\geq\log{x}$ for all $x\geq 1$).
Consequently, $C_1\geq C$, and by assuming that the units are
properly scaled such that $\bar{D}\geq1$, it follows that
$C\bar{D}\geq1$ and hence $C_1\bar{D}\geq1$.

Secondly, we prove
that $\gamma$, i.e., the optimal value in the solution of $C_1$,
satisfies $\gamma< e$. Consider the function
$f(\gamma)=\frac{\gamma^{a+1}}{a\log{\gamma}}$ with
$a=\frac{\alpha-1}{\alpha}$. If, by contradiction, $\gamma\geq e$,
then $f'(\gamma)> 0$ and consequently $f(\gamma)$ is
increasing on $[e,\infty)$. Since the function
$\frac{1}{C_1-\gamma\lambda}$ is also increasing in $\gamma$,
we get a contradiction that $\gamma$ is the optimal solution as
assumed, and hence $\gamma<e$.

Finally, consider the function
$g(a)=\frac{\gamma^a}{a\log{\gamma}}$ with
$a=\frac{\alpha-1}{\alpha}$. The previous property $\gamma<e$
implies that $g'(a)\leq 0$ and further that $g$ is non-increasing
in $a$ and hence in $\alpha$ as well. Since
$\frac{1}{(C_1\bar{D})^{\alpha-1}}$ is also non-increasing in
$\alpha$, we obtain that
\begin{eqnarray*}
&&\hspace{-0.5cm}\inf_{1<\gamma<\frac{C_1}{\lambda}}\left\{\frac{\gamma}{(C_1D)^{\alpha_1-1}\left(C_1-\gamma\lambda\right)}\frac{\gamma^{\frac{\alpha_1-1}{\alpha_1}}}{\log{\gamma^{\frac{\alpha_1-1}{\alpha_1}}}}\right\}\\
&&\geq\inf_{1<\gamma<\frac{C_1}{\lambda}}\left\{\frac{\gamma}{(C_1D)^{\alpha_2-1}\left(C_1-\gamma\lambda\right)}\frac{\gamma^{\frac{\alpha_2-1}{\alpha_2}}}{\log{\gamma^{\frac{\alpha_2-1}{\alpha_2}}}}\right\}~.
\end{eqnarray*}
Using the monotonicity in $C_1$ in the term inside the infimum, it
follows that $C_1\geq C_2$, where $C_2$ is the implicit solution
of Eq.~(\ref{eq:solHT}) for $\alpha=\alpha_1$. Therefore,
$C(\bar{D},\bar{\eps})$ is non-increasing in $\alpha$.

\emph{Scaling law:} To prove the scaling law, denote by
$C$ the implicit solution of the equation
\begin{equation}
\inf_{1<\gamma<\max\left\{\frac{C}{\lambda},e^{\frac{2\alpha}{\alpha-1}}\right\}}\left\{\frac{1}{C^{\alpha}}\frac{\gamma^{\frac{\alpha-1}{\alpha}}}{\log{\gamma^{\frac{\alpha-1}{\alpha}}}}\right\}=\bar{\eps}\bar{D}^{\alpha-1}~.\label{eq:relaxedIE}
\end{equation}
The LHS here was constructed by relaxing the function inside the
infimum in the LHS of Eq.~(\ref{eq:solHT}) and extending the range
of the infimum. This means that the implicit solution $C$ is
smaller than the implicit solution $C(\bar{D},\bar{\eps})$. The
function inside the infimum of the LHS of Eq.~(\ref{eq:relaxedIE})
is convex on the domain of $\gamma$ and attains its infimum at
$\gamma=e^{\frac{\alpha}{\alpha-1}}$. Solving for $C$ and using
that $C(\bar{D},\bar{\eps})\geq C$ proves the lower bound.

To prove the upper bound, let us fix $\alpha$, $\bar{D}_0$ and
$\bar{\eps}_0$, and denote by $C_0(\bar{D}_0,\eps_0)$ the
corresponding implicit solution. Using the monotonicity of the
implicit solution in $\bar{\eps}\bar{D}^{\alpha-1}$, as shown
above, it follows that
\begin{equation*}
C(\bar{D},\bar{\eps})\geq C_0(\bar{D}_0,\bar{\eps}_0)~,
\end{equation*}
where $\bar{\eps}\bar{D}^{\alpha-1}\leq \bar{\eps}_0\bar{D}_0^{\alpha-1}$. Fixing
$\gamma_0=\frac{C_0(\bar{D}_0,\bar{\eps}_0)+\lambda}{2\lambda}$, let
$C$ be the solution of the equation
\begin{equation}
\frac{\gamma_0}{C^{\alpha-1}\left(C-\gamma_0\lambda\right)}\frac{\gamma_0^{\frac{\alpha-1}{\alpha}}}{\log{\gamma_0^{\frac{\alpha-1}{\alpha}}}}=\bar{\eps}\bar{D}^{\alpha-1}~.\label{eq:pSc2}
\end{equation}
Because the range of $\gamma$ in the solution of
$C(\bar{D},\bar{\eps})$ includes $\gamma_0$, it follows that
$C(\bar{D},\bar{\eps})\leq C$. On the other hand, the LHS of
Eq.~(\ref{eq:pSc2}) satisfies
\begin{equation}
\frac{\gamma_0}{C^{\alpha-1}\left(C-\gamma_0\lambda\right)}\frac{\gamma_0^{\frac{\alpha-1}{\alpha}}}{\log{\gamma_0^{\frac{\alpha-1}{\alpha}}}}
\leq
\frac{K_0}{C^{\alpha}}~,\label{eq:pSc3}
\end{equation}
where
$K_0=\frac{\gamma_0C_0(\bar{D}_0,\bar{\eps}_0)}{C_0(\bar{D}_0,\bar{\eps}_0)-\gamma_0\lambda}\frac{\gamma_0^{\frac{\alpha-1}{\alpha}}}{\log{\gamma_0^{\frac{\alpha-1}{\alpha}}}}$.
Here we used that $C\geq C_0(\bar{D}_0,\bar{\eps}_0)$ (note that
we showed before that $C\geq C(\bar{D},\bar{\eps})$ and
$C(\bar{D},\bar{\eps})\geq C_0(\bar{D}_0,\bar{\eps}_0)$). Finally,
combining Eqs.~(\ref{eq:pSc2}) and (\ref{eq:pSc3}) we immediately
get the scaling law
$C={\mathcal{O}}\left(\left(\frac{1}{\bar{\eps}\bar{D}^{\alpha-1}}\right)^{\frac{1}{\alpha}}\right)$,
and since $C(\bar{D},\bar{\eps})\leq C$, the proof is complete.
\end{proof}

\section{Conclusion}
\label{s.conclusion}

Our goal in this paper is to provide new insight into the debate about the potential of dynamic resizing in data centers.  Clearly, there are many facets of this issue relating to the engineering, algorithmic, and reliability challenges involved in dynamic resizing which we have ignored in this paper.  These are all important issues when trying to \emph{realize the potential} of dynamic resizing. But, the point we have made in this paper is that when \emph{quantifying the potential} of dynamic resizing it is of primary importance to understand the joint impact of workload and SLA characteristics.

To make this point, we have presented a new model that, for the first time, captures the impact of SLA characteristics in addition to both slow time-scale non-stationarities in the workload and fast time-scale burstiness in the workload.   This model allows us to provide the first study of dynamic resizing that captures both the stochastic burstiness and diurnal non-stationarities of real workloads. Within this model, we have provided both trace-based numerical case studies and analytical results.  Perhaps most tellingly, our results highlight that even when two of SLA, peak-to-mean ratio, and burstiness are fixed, the other one can be chosen to ensure that there either are or are not significant savings possible via dynamic resizing. Figures \ref{fig:whenvaluablebysavings}-\ref{fig:whenvaluablebyepsilon} illustrate how dependent the potential of dynamic resizing is on these three parameters.  These figures highlight that a precursor to any debate about the value of dynamic resizing must be an understanding of the workload characteristics expected and the SLA desired.  Then, one can begin to discuss whether this potential is obtainable.

Future work on this topic includes providing a more detailed study of how other important factors affect the potential of dynamic resizing, e.g., storage issues, reliability issues, and the availability of renewable energy.  Note that provisioning capacity to take advantage of renewable energy when it is available is an important benefit of dynamic resizing that we have not considered at all in the current paper.

\section*{Acknowledgment}
This research is supported by the NSF grant of China (No. 61303058), the Strategic Priority Research Program of the Chinese Academy of Sciences (No. XDA06010600), the 973 Program of China (No. 2010CB328105), and NSF grant CNS 0846025 and DoE grant DE-EE0002890.

%% The Appendices part is started with the command \appendix;
%% appendix sections are then done as normal sections
%% \appendix

%% \section{}
%% \label{}

%% References
%%
%% Following citation commands can be used in the body text:
%% Usage of \cite is as follows:
%%   \cite{key}         ==>>  [#]
%%   \cite[chap. 2]{key} ==>> [#, chap. 2]
%%

%% References with bibTeX database:

\bibliographystyle{elsarticle-num}
\bibliography{ref}

%% Authors are advised to submit their bibtex database files. They are
%% requested to list a bibtex style file in the manuscript if they do
%% not want to use elsarticle-num.bst.

%% References without bibTeX database:

% \begin{thebibliography}{00}

%% \bibitem must have the following form:
%%   \bibitem{key}...
%%

% \bibitem{}

% \end{thebibliography}

\end{document}